\documentclass[a4paper]{article}
\usepackage[english]{babel}
\usepackage[utf8x]{inputenc}
\usepackage[T1]{fontenc}
\usepackage[a4paper,top=3cm,bottom=2cm,left=2.7cm,right=2.7cm,marginparwidth=1.75cm]{geometry}
\usepackage{comment}
\usepackage{authblk}
\usepackage[misc]{ifsym}
\usepackage{indentfirst}
\usepackage[english]{babel}
\usepackage[utf8x]{inputenc}
\usepackage{amsmath,amsfonts,amssymb,amsthm}
\usepackage[english]{babel}
\usepackage{graphicx}
\usepackage{hyperref}
\usepackage{latexsym}
\usepackage{ifthen}
\usepackage{wrapfig}
\usepackage{mathtools}
\usepackage{braket}
\usepackage{algorithm}
\usepackage{algpseudocode}
\usepackage{verbatim}
\usepackage[qm]{qcircuit}
\usepackage{todonotes}
\usepackage{multirow}
\usepackage[toc,page]{appendix}
\usepackage{listings}
\usepackage{xcolor}

\makeatletter

\lstdefinelanguage{AMPL}{keywords={set,param,var,arc,integer,minimize,maximize,subject,to,node,sum,in,Current,complements,integer,solve_result_num,IN,contains,less,suffix,INOUT,default,logical,sum,Infinity,dimen,max,symbolic
,Initial,div,min,table,LOCAL,else,option,then,OUT,environ,setof ,union,all,exists,shell_exitcodeuntil,binary,forall,solve_exitcodewhile ,by,if,solve_messagewithin,check,in,solve_result
},sensitive=true,comment=[l]{\#}}

\DeclarePairedDelimiter{\ceil}{\lceil}{\rceil}
\DeclarePairedDelimiter\abs{\lvert}{\rvert}

\DeclarePairedDelimiter\Ceil\lceil\rceil

\theoremstyle{plain}
\newtheorem{thm}{Theorem}[section]

\newtheorem{prop}[thm]{Proposition}
\newtheorem{cor}[thm]{Corollary}
\newtheorem{defi}[thm]{Definition}

\newcommand{\NOT}{\operatorname{NOT}}

\title{Quantum computing from a mathematical perspective:\\a description of the quantum circuit model}
\author{J. Ossorio-Castillo\thanks{Email: joaquin@mestrelab.com (\Letter).} , Jos\'e M. Tornero\thanks{Email: tornero@us.es.}}
\date{April 4th, 2025} 

\begin{document}
\maketitle

\begin{abstract}
This paper is an essentially self--contained and rigorous description of the fundamental principles of quantum computing from a mathematical perspective. It is intended to help mathematicians who want to get a grasp of this quickly growing discipline and find themselves taken aback by the language gap between mathematics and the pioneering fields on the matter: computer science and quantum physics. 

A first version of this paper was originally published at arXiv in October 2018. Minor corrections added on Spring 2025.
\end{abstract}

\section{Historical precedents of Quantum Computation}

The origins of quantum computation date back to 1980, when Paul Benioff described a computing model defined by quantum mechanical Hamiltonians \cite{benioff1980computer}. Later that year, Yuri Manin gave a first idea on how to simulate a quantum system with a computer governed by quantum mechanics \cite{manin1980computable}. Both of them laid the groundwork for two of the basic components of quantum computing: quantum Turing machines and quantum computers \cite{nielsen2002quantum}.

Two years later, Richard Feynman explored the problems of simulating physics with a classical computer in one of his most seminal papers \cite{feynman1982simulating}, and introduced independently a quantum model of computation. He stated that, being the world quantum mechanical, the difficulty for replicating exactly the behavior of nature is related to the problem of simulating quantum physics. This way, the most important rule defined by Feynman deals with the computational complexity at the time of efficiently simulating a quantum system. If one doubles the dimension of the system, it would be ideal that the size of the computational resources needed for this task also double in the worst case, instead of experiencing an exponential growth.

Feynman also stated the underlying limitations that appear when it comes to simulate the probabilities of a physical system. Instead of calculating the probabilities of such a system, which he proved to be impossible, he proposed that the computer itself should have a probabilistic nature. To this new kind of machine he gave the name of quantum computer, and stated that it had a distinct essence than the well-known Turing machines. He also noted that with one of them it should be possible to simulate correctly any quantum system, and the physical world itself. Feynman asked himself if it would be possible to define a universal quantum computer, capable of modeling all possible quantum systems and detached from the possible problems that originate from its physical implementation, in the same way that a classical one is.

\newpage

\section{Quantum Turing Machines}

Although the credit for introducing the concept of a universal quantum computer usually goes to Richard Feynman, it was David Deutsch the first to properly describe, generalize and formalize it \cite{deutsch1985quantum}. Supported in the works by Feynman, Manin and Benioff, he also introduced the concept of quantum Turing machine (QTM), which we proceed to explain succintly.

A quantum Turing machine $\mathcal{M}$, or {\em QTM}, is defined  (as are classical Turing machines) by a triple $\mathcal{M} = (S, \Sigma, \delta)$, but the usual set of states $S$ of a Turing machine is replaced by (some) vectors of a Hilbert space, the alphabet $\Sigma$ is finite, and the transition function $\delta$ is substituted for a set of unitary transformations which are automorphisms of the Hilbert space.

This definition is rather informal and leaves out many important details. In fact, the study of quantum Turing machines is quite intricate. Fortunately, an equivalent and much more friendly model of computation called the quantum circuit model exists, and will be explained later on in this paper along with many of its details. Nevertheless, the reader interested in the complete and original definition of quantum Turing machines and all of its characteristics, can refer to the seminal papers where it was first outlined and formalized \cite{deutsch1985quantum,deutsch1989quantum,bernstein1997quantum}. A quantum Turing machine can also be seen as a probabilistic Turing machine that obeys the rules of quantum probability instead of classical probability \cite{simon1997power}.

\ 

The counterpart of the P class is given in the QTM context by the complexity class BQP, which stands for Bounded-error Quantum Polynomial-time. It contains all decision problems that can be solved in polynomial time by a quantum Turing machine with error probability bounded by 1/3 for all inputs. 

The latter class is usually taken as a reference for representing the power of quantum computers. Thanks to \cite{bernstein1997quantum} and \cite{deutsch1985quantum}, we already know that BQP $\subseteq$ PSPACE and it is trivial to see that P $\subseteq$ BQP. However, at the present time there is no known relationship between NP and BQP, except that P is inside their intersection. There is a strong belief that NP $\nsubseteq$ BQP; consequently, a polynomial-time quantum algorithm for a NP-complete problem would be surprising, as it would violate this conjecture. A problem that is not known to be in P is the factoring problem, but we already know thanks to Peter W. Shor a polynomial-time algorithm for this problem that runs on a quantum computer \cite{shor1994algorithms}. This algorithm, among many others, will be thoroughly explained in a next section. 

\ 

\begin{wrapfigure}{r}{0.3\textwidth}\label{fig::bqpdiagram}
\includegraphics[width=0.3\textwidth]{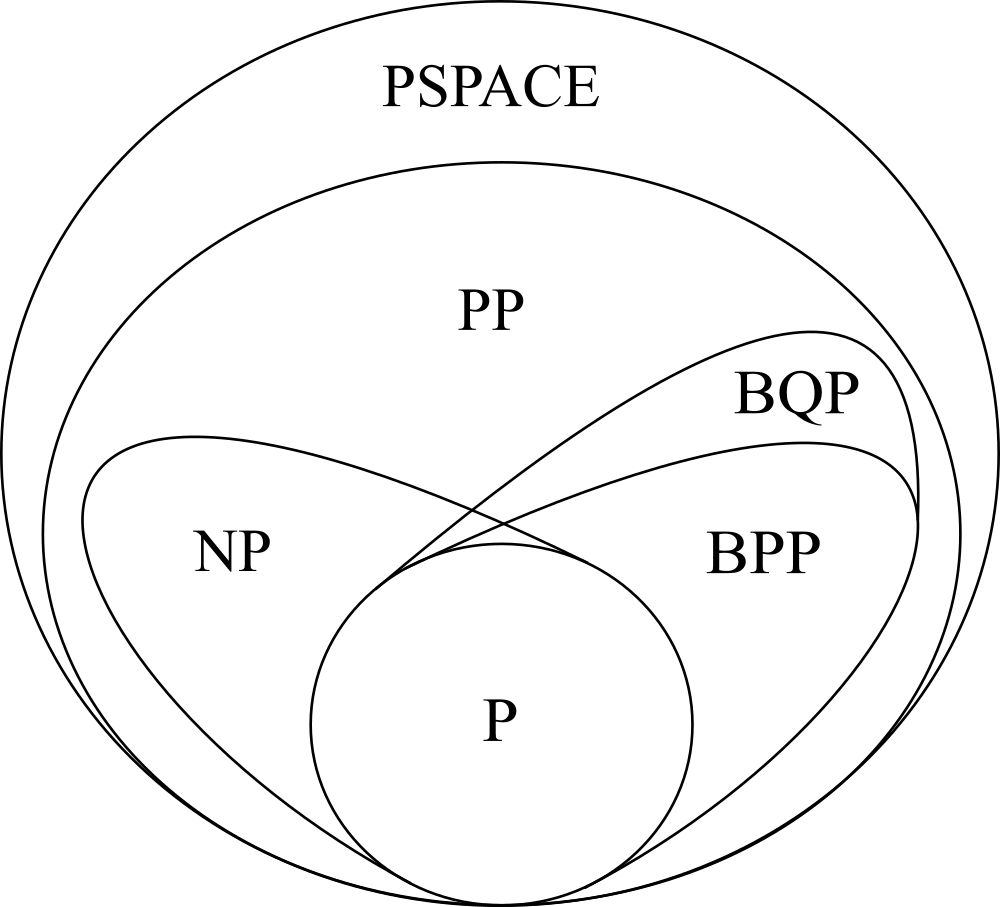}
\centering
\caption{Complexity classes (incl. BQP)}
\end{wrapfigure}

It can be deduced from the definition of the complexity class BQP that the inner nature of quantum computation is probabilistic. In order to measure the performance of a quantum algorithm that solves a certain problem, we do not usually take into account the time needed for obtaining the solution to that problem. What we do is to study the relationship between the probability of obtaining a correct solution and the computation time. In order for an quantum algorithm to be considered efficient, it must return a correct solution in polynomial time with a probability of at least 2/3. For a groundbreaking study on the algorithmic limitations of quantum computing, we refer to \cite{bennett1997strengths}.

\

We must mention here that the quantum circuit model is not the only quantum computation paradigm that is currently being developed. There exists a completely different approach to exploiting the possibilities of quantum physics in computation, called the {\em adiabatic model}, which is equivalent also to a QTM \cite{aharonov2008adiabatic}. We will not treat this matter here, as it needs a substantially diverse approach (and this paper is already long enough as it is).

\newpage

\section{Quantum Bits and Quantum Entanglement}\label{sec:qubits}

In classical computation, the basic unit of information is the bit (a portmanteau of binary digit). A bit can only be in one of its two possible states, and may therefore be physically implemented with a two-state device. This pair of values is commonly represented with $0$ and $1$. On the other hand, we have an analogous concept in quantum computation: the qubit (short for quantum bit) which is a mathematical representation of a two-state quantum-mechanical system.

We will work in the Hilbert space $\mathbb C^2$, with the usual scalar product. In the quantum parlance, vectors are, however, written in a different way.

\begin{defi}
The vectors 
$$ 
\ket{0} = \left[\begin{array}{ r }
1 \\
0
\end{array}\right] \text{ and } \ket{1} = \left[\begin{array}{ r }
0 \\
1
\end{array}\right]
$$
are called the \textbf{basis states} of a \textbf{quantum bit}.
\end{defi} 

So far we have only two states and it does not seem to far apart from the bit situation (funny notations aside). However, there is a difference between bits and qubits: a qubit can also be in a state other than $\ket{0}$ and $\ket{1}$. Its generic state is, in fact, a linear combination over the complex numbers of both basis states.

\begin{defi}
A \textbf{pure qubit state} $\ket{\psi}$ is a unit vector which is a linear combination of the basis states,
$$
\ket{\psi} = \alpha\ket{0} + \beta\ket{1} = \left[\begin{array}{ r }
\alpha \\
\beta
\end{array}\right]
$$ 
where the coefficients $\alpha , \beta \in \mathbb{C}$ are called the amplitudes of the state.
\end{defi}

The fact that $\ket{\psi}$ is a unit vector means of course that the constraint $\abs{\alpha}^2 + \abs{\beta}^2 = 1$ holds.

\paragraph{Remark:} The notation we have just introduced, $\ket{\psi}$, termed \textit{ket}, for describing a quantum state. This notation is part of the bra-ket notation, also named Dirac notation in honor of Paul Dirac, who first introduced it in 1939 \cite{dirac1939new}. Alternatively, we will also use $\bra{\psi}$, called \textit{bra}, to describe $\ket{\psi}^*$, the Hermitian conjugate of $\ket{\psi}$.

\

Thus, $\ket{0}$ and $\ket{1}$ form an orthonormal $\mathbb{C}$-basis of $\mathbb{C}^2$. From now on, the basis formed by $\ket{0}$ and $\ket{1}$ will be called the computational basis of a qubit. Nevertheless, there are other commonly used basis for the states of a quantum bit. An example that will come in handy later, known as the Hadamard basis, is defined by 
$$
\ket{+} = \dfrac{1}{\sqrt{2}} (\ket{0} + \ket{1}) = \dfrac{1}{\sqrt{2}} \left[\begin{array}{ r }
1 \\ 1 \end{array}\right]
$$ 
and 
$$
\ket{-} = \dfrac{1}{\sqrt{2}} (\ket{0} - \ket{1}) = \dfrac{1}{\sqrt{2}} \left[\begin{array}{ r } 1 \\ -1 \end{array}\right].
$$ 

It is easy to see that $\ket{+}$ and $\ket{-}$ also form an orthonormal $\mathbb{C}$-basis of $\mathbb{C}^2$, and that our generic qubit $\ket{\psi} = \alpha\ket{0} + \beta\ket{1}$ can be seen as
$$
\ket{\psi} = \dfrac{\alpha + \beta}{\sqrt{2}} \ket{+} + \dfrac{\alpha - \beta}{\sqrt{2}} \ket{-}.
$$

We know that qubits exist in nature thanks to the Stern-Gerlach experiment, first conducted by Otto Stern and Walther Gerlach in 1922 \cite{gerlach1922experimentelle}.

\paragraph{Remark:} One of the key features that makes a quantum computer differ dramatically from its classical counterpart is the process of measuring the state of a quantum bit. A measurement, also called observation, of a generic single-qubit state $\ket{\psi} = \alpha\ket{0} + \beta\ket{1}$ is a physical procedure that yields a result from the orthonormal basis, depending on the values of $\alpha$ and $\beta$. This dependence is modelled as a Bernoulli distribution: the probability that the measurement gives $\ket{0}$ as a result is $|\alpha|^2$ and, obviously, the probability that the measurement yields $\ket{1}$  is $|\beta|^2$.

However, unlike the classical case, the measurement process inevitably disturbs the qubit $\ket{\psi}$, forcing it to collapse to either $\ket{0}$ or $\ket{1}$ and thus generally making impossible the task of finding out the actual values of $\alpha$ and $\beta$. This collapse to either $\ket{0}$ or $\ket{1}$ is then non-deterministic and non-reversible, and this is a fundamental feature of quantum computation. It will be shown in Section \ref{sec:quantumcircuits} how to change these probabilities without violating the unitary constraint.  

Therefore, in short, after measuring, one might think of a qubit as a non-deterministic bit. It can all take two possible values, all the same, but which of them does one actually get depends on a probability distribution. And, most importantly, the data given by these probabilities vanish once the qubit is observed.

\paragraph{Remark:} A possible geometrical representation of the states of a single-qubit system, known as the Bloch sphere \cite{bloch1946nuclear,nielsen2002quantum} (see Figure \ref{fig::bloch}), leans on the following interpretation. The amplitudes $\alpha$ and $\beta$ are not interesting by themselves. For one thing, their moduli is who characterizes the probability distribution that actually matters. Therefore, two qubits which feature the same distribution are, in fact, computationally indistinguishable \textit{after measurement}. 

Building upon this, we may choose a representative for all such qubits: for instance we may force $\alpha$ to be a real number, and it is straightforward that, given a qubit $\ket{\psi}$ there is only one qubit $\ket{\psi_0}$ with the form
$$
\ket{\psi_0} = e^{i\varphi} \ket{\psi} = \cos \left( \frac{\theta}{2} \right) \ket{0} + e^{i \phi} \sin \left( \frac{\theta}{2} \right) \ket{1}
$$
with $\theta \in [0,\pi]$ and $\phi \in [0,2\pi)$. 

Choosing such a representation, we have that a generic qubit can be represented uniquely as a point $(\theta, \phi)$ of the unit 2-sphere, with the north and south poles typically chosen to correspond to the standard basis vectors $\ket{0}$ and $\ket{1}$ as indicated in Figure 2.  

\begin{figure}\label{fig::bloch}
\includegraphics[height=6cm]{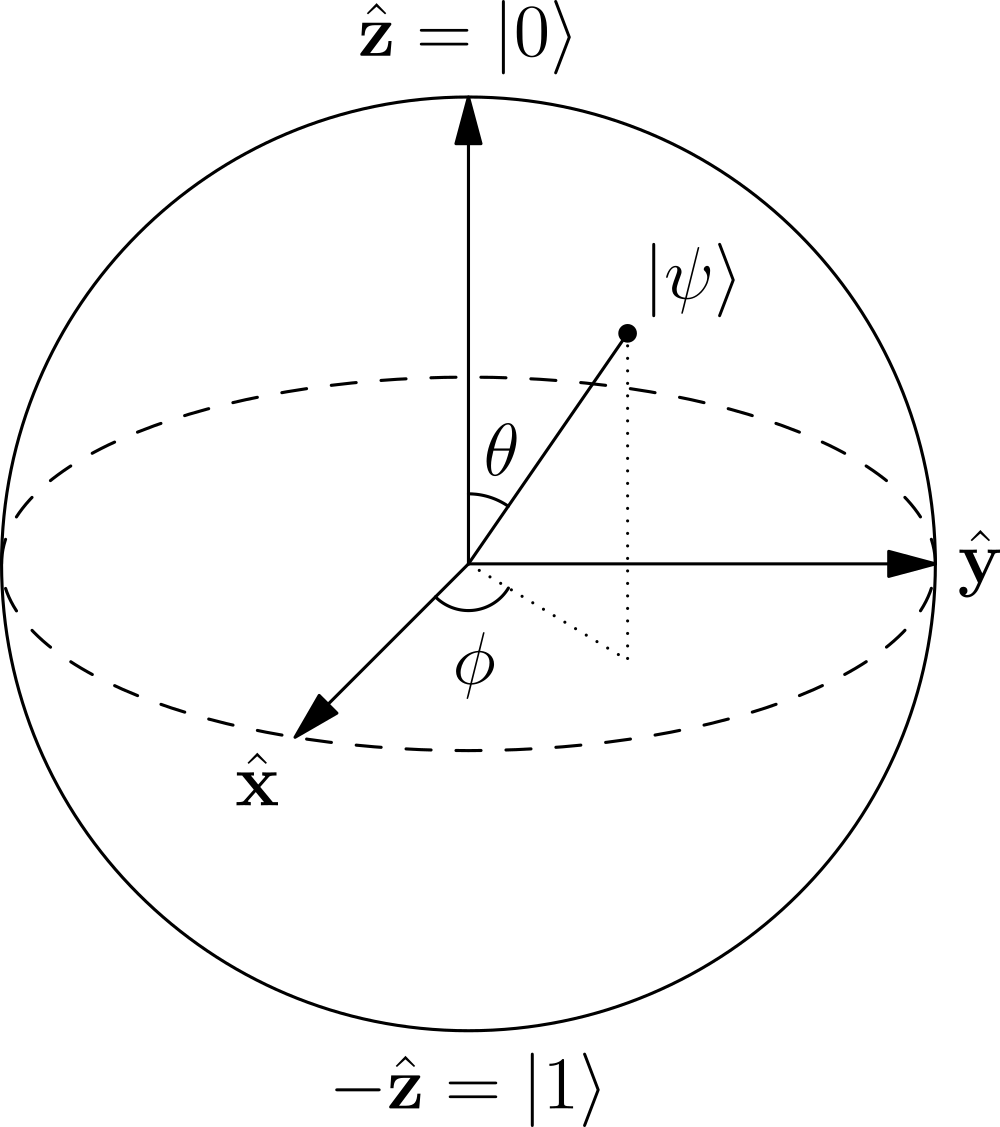}
\centering
\caption{Bloch sphere\protect\footnotemark}
\end{figure}

\footnotetext{Credit: \href{https://upload.wikimedia.org/wikipedia/commons/f/f4/Bloch_Sphere.svg}{Glosser.ca} (\href{https://creativecommons.org/licenses/by-sa/3.0/}{CC BY-SA 3.0})}

The Bloch sphere is useful as a way of depicting transformations of a single qubit. As we will see later on, the most important transformations can be broken down to (essentially) rotations in the Bloch sphere. We will look at this matter in the next section. But first, let us explain some concepts needed to understand how multiple-qubit systems behave.

\ 

Going from a single qubit to a multiple-qubit system should obviously involve considering a Hilbert vector space which stems from combining multiple copies of $\mathbb C^2$. The chosen way for doing that is using the tensor product.

\begin{defi}
Let $V$ and $W$ be vector spaces of dimensions $n$ and $m$ respectively. The \textbf{tensor product of $V$ and $W$}, denoted by $V \otimes W$, is an $nm$-dimensional vector space whose elements are linear combinations of the symbols $v \otimes w$ satisfying the subsequent properties:
\begin{itemize}
\item $\alpha (v \otimes w) = (\alpha v) \otimes w = v \otimes (\alpha w)$
\item $(v_1 + v_2) \otimes w = (v_1 \otimes w) + (v_2 \otimes w)$
\item $v \otimes (w_1 + w_2) = (v \otimes w_1) + (v \otimes w_2)$
\end{itemize}
where $\alpha \in \mathbb{C}$, $v,v_1,v_2 \in V$ and $w,w_1,w_2 \in W$.
\end{defi}

\paragraph{Remark:} Given two bases for $V$ and $W$, say ${\mathcal B}_V = \{ v_1,\ldots,v_n \}$ and ${\mathcal B}_W = \{ w_1,\ldots,w_m \}$, it is well--known that the set
$$
{\mathcal B}_{V \otimes W} = \{ v_i \otimes w_j \; | \; 1 \leq i \leq n, \; 1 \leq j \leq m \}
$$
is a basis of $V \otimes W$.

A related definition that will come to use in Section \ref{sec:quantumcircuits} is the concept of tensor product between linear operators.

\begin{defi}
Let $A$ and $B$ be linear operators defined on $V$ and $W$ respectively, then the \textbf{linear operator $A \otimes B$ operating on $V \otimes W$} is defined as 
$$
(A \otimes B) (v \otimes w) = Av \otimes Bw
$$ 
with $v \in V$ and $w \in W$.
\end{defi}

If $A$ and $B$ are $n \times n$ and $m \times m$ matrices respectively that correspond to the matrix representations of the linear operators $A$ and $B$ with respect to the canonical base, the linear operator $A \otimes B$ (called the tensor product, or the Kronecker product of $A$ and $B$) has the following matrix representation with respect to the canonical base:
$$
A \otimes B
=
\begin{bmatrix}
    a_{11}B & a_{12}B & \cdots & a_{1n}B \\
    a_{21}B & a_{22}B & \cdots & a_{2n}B \\
    \vdots & \vdots & \ddots & \vdots \\
    a_{n1}B & a_{n2}B & \cdots & a_{nn}B \\
\end{bmatrix}
$$ 
As expected the matrix representation of $A \otimes B$ has dimension $nm \times nm$. As it happens with the usual product, this operation is not commutative. A common notation for the Kronecker product of $l$ copies of a matrix $A$ is $A^{\otimes l}$.

\paragraph{Example:} By means of showing how the matrix representations of the Kronecker products of linear operators are calculated, let 
$$
A =
\begin{bmatrix*}[r]
    1 & -1 \\
    -2 & 0
\end{bmatrix*}
\text{ and }
B = 
\begin{bmatrix*}[r]
    1 & 0 & 0 \\
    0 & 2 & 0 \\
    0 & 0 & 3
\end{bmatrix*}
$$ 
be two linear operations defined on $\mathbb{R}^2$ and $\mathbb{R}^3$ respectively. Then, their tensor product is calculated as follows:
$$
A \otimes B =
\begin{bmatrix*}[r]
    1 & -1 \\
    -2 & 0
\end{bmatrix*}
\otimes
\begin{bmatrix*}[r]
    1 & 0 & 0 \\
    0 & 2 & 0 \\
    0 & 0 & 3
\end{bmatrix*}
=
\begin{bmatrix*}[r]
    1 & 0 & 0 & -1 & 0 & 0 \\
    0 & 2 & 0 & 0 & -2 & 0 \\
    0 & 0 & 3 & 0 & 0 & -3 \\
    -2 & 0 & 0 & 0 & 0 & 0 \\
    0 & -4 & 0 & 0 & 0 & 0 \\
    0 & 0 & -6 & 0 & 0 & 0 \\
\end{bmatrix*}
$$ 
where $A \otimes B$ is a linear operator defined on $\mathbb{R}^6$.

\ 

Of course, the previous definitions are extended in the direct way to finite tensor products of spaces and operators. In particular, note that the tensor product of unit vectors is again a unit vector.

Also, the tensor product we have thus defined can also be extended to vectors and non-square matrices, and will be useful at the time of calculating the basis states of a quantum system with more than one qubit and representing it as a vector in $\mathbb C^l$, for some $l \in \mathbb N$. 

\paragraph{Example:} For instance, if $\ket{0}$ and $\ket{1}$ are the basis states of a quantum bit, the tensor product $\ket{1} \otimes \ket{0}$ will be given by:
$$
\ket{1} \otimes \ket{0} = \left[\begin{array}{ r }
0 \\
1
\end{array}\right]
\otimes 
\left[\begin{array}{ r }
1 \\
0
\end{array}\right] = 
\left[\begin{array}{ r }
0 \\
0 \\
1 \\
0 \\
\end{array}\right]
$$

Before continuing with the possible states of a multiple qubit system, let us introduce some notation. In our set-up, the corresponding bases will be of use for representing integers. In a classical computer, we represent an integer $a \in \mathbb{Z}_{\geq 0}$ such that $a < 2^n$ (i.e., such that it can be described with $n$ bits) with the base-2 numeral system:
$$ 
a = \sum^{n-1}_{l=0} a_l 2^l
$$ 
where $a_l \in \{0,1\}$ are the binary digits of $a$. In a quantum computer, we can also represent an integer $a < 2^n$ with $n$ \textit{qubits} as follows:
$$ 
\ket{a}_n = \ket{a_{n-1} \cdot\cdot\cdot a_1 a_0} = \bigotimes^{n-1}_{l=0} \ket{a_l}
$$

Thus, for example, number 29 can be represented with 5 qubits (as $29 < 2^5)$ like this:
$$ 
\ket{29}_5 = \ket{11101} = \ket{1} \otimes \ket{1} \otimes \ket{1} \otimes \ket{0} \otimes \ket{1}.
$$

In this way, integers are always represented by elements of the basis which is obtained from tensor products of the single-qubit computational bases. This basis will be subsequently called the computational basis itself.

\paragraph{Notation:} From now on, the notation $\ket{\psi}_n$ will imply that we are describing a $n$-qubit system (with $n \geq 2$) instead of a single-qubit one, which will remain to be indicated with the absence of a subindex. We will also make use sometimes of the notation $\ket{uv}$ to describe the tensor product $\ket{u} \otimes \ket{v}$ of two basis states. Now that we know what $\ket{\psi}_n$ and $\ket{a}_n$ really mean, we are finally in the position to begin studying the possible states of a multiple qubit system which is, once again essentially, a unitary vector in the corresponding Hilbert space.

\begin{defi}
The \textbf{state $\ket{\psi}_n$ of a generic $n$-qubit system} is a superposition (that is, a linear combination) of the $2^n$ states of the computational basis $\ket{0}_n, \ket{1}_n, \ldots, \ket{2^n-1}_n$ with modulus $1$. In particular, 
$$
\ket{\psi}_n = \sum^{2^n-1}_{j=0} \alpha_j\ket{j}_n,
$$ 
with amplitudes $\alpha_j \in \mathbb{C}$ constrained to 
$$ 
\sum^{2^n -1}_{j=0} \abs{\alpha_j}^2 = 1.
$$
\end{defi}

This can be seen as an obvious advantage with respect to classical computation. In a conventional computer we can store one and only one integer between $0$ and $2^n-1$ inside a $n$-bit register, which can be seen as a discrete probability distribution between all possible integers where the integer we have stored has probability 1 and the rest have 0. In a quantum register, the probability can be distributed between all integers from $0$ to $2^n-1$, instead of having just one possibility when it comes to read the register. Even more, if we are to simulate this quantum behavior with a classical computer, we would need $2^n$ registers of $n$ bits, instead of a single $n$-qubit register as in the quantum case. This is precisely one of the benefits of quantum computing that Richard Feynman foretold in his paper \cite{feynman1982simulating}. 

\paragraph{Example:} Let us have a look at the simplest case. The basis states of a two-qubit system are the tensor products of the basis states of a single-qubit system:
$$
\ket{0}_2 = \ket{00} = \ket{0} \otimes \ket{0} = \left[\begin{array}{ r }
1 \\
0
\end{array}\right] \otimes \left[\begin{array}{ r }
1 \\
0
\end{array}\right] = \left[\begin{array}{ r }
1 \\
0 \\
0 \\
0 \\
\end{array}\right], \quad 
\ket{1}_2 = \ket{01} = \ket{0} \otimes \ket{1} = \left[\begin{array}{ r }
1 \\
0
\end{array}\right] \otimes \left[\begin{array}{ r }
0 \\
1
\end{array}\right] = \left[\begin{array}{ r }
0 \\
1 \\
0 \\
0 \\
\end{array}\right],
$$
$$
\ket{2}_2 = \ket{10} = \ket{1} \otimes \ket{0} = \left[\begin{array}{ r }
0 \\
1
\end{array}\right] \otimes \left[\begin{array}{ r }
1 \\
0
\end{array}\right] = \left[\begin{array}{ r }
0 \\
0 \\
1 \\
0 \\
\end{array}\right], \quad
\ket{3}_2 = \ket{11} = \ket{1} \otimes \ket{1} = \left[\begin{array}{ r }
0 \\
1
\end{array}\right] \otimes \left[\begin{array}{ r }
0 \\
1
\end{array}\right] = \left[\begin{array}{ r }
0 \\
0 \\
0 \\
1 \\
\end{array}\right].
$$

And the generic state of two different single-qubit systems, described independently, can be represented as
$$ 
\ket{\psi_0} = \alpha\ket{0} + \beta\ket{1} = \alpha \left[\begin{array}{ r }
1 \\
0
\end{array}\right] + \beta  \left[\begin{array}{ r }
0 \\
1
\end{array}\right] = \left[\begin{array}{ r }
\alpha \\
\beta
\end{array}\right]
$$
and 
$$ 
\ket{\psi_1} = \gamma\ket{0} + \delta\ket{1} = \gamma \left[\begin{array}{ r }
1 \\
0
\end{array}\right] + \delta  \left[\begin{array}{ r }
0 \\
1
\end{array}\right] = \left[\begin{array}{ r }
\gamma \\
\delta
\end{array}\right],
$$ 
where $\abs{\alpha}^2 + \abs{\beta}^2 = 1$ and $\abs{\gamma}^2 + \abs{\delta}^2 = 1$. This means that the state of this 2-qubit system which arise from them should be described as the tensor product of both:
$$ 
\ket{\psi_0} \otimes \ket{\psi_1} =  \left[\begin{array}{ r }
\alpha \\
\beta
\end{array}\right] \otimes \left[\begin{array}{ r }
\gamma \\
\delta
\end{array}\right] = \left[\begin{array}{ r }
\alpha\gamma \\
\alpha\delta \\
\beta\gamma \\
\beta\delta
\end{array}\right].
$$

On the other hand, if we want to describe a generic 2-qubit system $\ket{\psi}_2$ with the basis states defined above, we would have
$$ 
\ket{\psi}_2 = \alpha_0
\left[\begin{array}{ r }
1 \\
0 \\
0 \\
0
\end{array}\right]
+ \alpha_1
\left[\begin{array}{ r }
0 \\
1 \\
0 \\
0
\end{array}\right]
+ \alpha_2
\left[\begin{array}{ r }
0 \\
0 \\
1 \\
0
\end{array}\right]
+ \alpha_3\left[\begin{array}{ r }
0 \\
0 \\
0 \\
1
\end{array}\right] = \left[\begin{array}{ r }
\alpha_0 \\
\alpha_1 \\
\alpha_2 \\
\alpha_3
\end{array}\right],
$$ 
where
$$ 
\abs{\alpha_0}^2 + \abs{\alpha_1}^2 + \abs{\alpha_2}^2 + \abs{\alpha_3}^2 = 1
$$ 
must hold (remember that any quantum system must be described as a unit vector). \\

Note that if our generic two-qubit system described by $\ket{\psi}_2$ is to be decomposed in two single-qubit states (i.e., $\ket{\psi}_2 = \ket{\psi_0} \otimes \ket{\psi_1}$), then 
$$
\alpha_0 = \alpha \gamma, \;\; \alpha_1 = \alpha \delta, \;\; 
\alpha_2 = \beta \gamma \mbox{ and } \alpha_3 = \beta \delta.
$$

It is easy to see that the equality $\alpha_0 \alpha_3 = \alpha_1 \alpha_2$ is imposed; however, it is clear that this condition does not necessarily holds in a two-qubit generic state.

This is the mathematical counterpart of a well-known physical phenomenon called quantum entanglement which implies that the quantum state of each one of the particles of a two-qubit system may not be described independently. This leads us to the subsequent definition:

\begin{defi}
An $n$-qubit general state $\ket{\psi}_n$ is called \textbf{entangled} if there does not exist $n$ one-qubit states $\ket{\psi_0}, \ldots, \ket{\psi_n}$ such that 
$$
\ket{\psi}_n = \ket{\psi_0} \otimes \ldots \otimes \ket{\psi_n}.
$$

If a state is not entangled we call it \textbf{separable}.
\end{defi}

\paragraph{Remark:}  Quantum entanglement was first observed in nature in 1935, and in early days it was known as the Einstein–Podolsky–Rosen paradox. It was first studied by Albert Einstein and his colleagues Boris Podolsky and Nathan Rosen \cite{einstein1935can}, and later by Erwin Schrödinger \cite{schrodinger1935discussion}. The role and importance of quantum entanglement in quantum algorithms operating on pure states and in quantum computational speed-up was extensively discussed by Richard Jozsa and Noah Linden in \cite{jozsa2003role}. 

In particular, it is shown in the above reference that a quantum computer which does not take advantage of the quantum entanglement is not too far apart from a classical computer. In fact, the most interesting result that links quantum entanglement and quantum computing performance over classical computation is the following, which can be found in \cite{gottesman1998heisenberg}:

\begin{thm}
{\em \textbf{(Gottesman–Knill)}} A quantum algorithm that starts in a computational basis state and does not feature quantum entanglement can be simulated in polynomial time by a probabilistic classical computer.
\end{thm}

Therefore, it is precisely quantum entanglement what might give quantum computing a head start, compared to classical computing. One of the many challenges on the hardware side is precisely to create a stable enough environment for entanglement, which is a very delicate and fragile phenomenon. 

\paragraph{Remark:}  Analogously to the single-qubit case, observing an $n$-qubit system unavoidably interferes with $\ket{\psi}_n$ and impels it to collapse in one of the vectors of the computational basis (i.e., in $\ket{j}_n$ with $0 \leq j < 2^n$). This collapse is again non-deterministic and is governed by the discrete  probability distribution given by $\abs{\alpha_j}^2$. Thus, all the information that may have been stored in the amplitudes $\alpha_j$ is inevitably lost after the measurement process. 

\paragraph{Example:} By way of illustration, let us suppose that we have the following 3-qubit quantum system:
$$
\ket{\psi}_3 = \dfrac{1}{2} \ket{1}_3 + \dfrac{1}{2} \ket{3}_3 + \dfrac{1}{2} \ket{5}_3 + \dfrac{1}{2} \ket{7}_3.
$$ 

Then, if we measure this system, we will obtain with identical probability one of these possible outcomes: 1, 3, 5 or 7. Additionally, it is interesting to see the behavior of a quantum system if, rather than measuring all qubits at once, we measure them one by one. Our previous quantum system can be seen as 
$$
\ket{\psi}_3 = \dfrac{1}{2} \ket{001} + \dfrac{1}{2} \ket{011} + \dfrac{1}{2} \ket{101} + \dfrac{1}{2} \ket{111}.
$$ 

But also as 
$$
\ket{\psi}_3 = \dfrac{1}{\sqrt{2}} \ket{0} \otimes \Big(\dfrac{1}{\sqrt{2}}\ket{01} + \dfrac{1}{\sqrt{2}}\ket{11}\Big) + \dfrac{1}{\sqrt{2}} \ket{1} \otimes \Big( \dfrac{1}{\sqrt{2}} \ket{01} + \dfrac{1}{\sqrt{2}} \ket{11} \Big) 
$$ 
or as 
$$
\ket{\psi}_3 = \Big( \dfrac{1}{\sqrt{2}} \ket{0} + \dfrac{1}{\sqrt{2}} \ket{1} \Big) \otimes \Big( \dfrac{1}{\sqrt{2}} \ket{01} + \dfrac{1}{\sqrt{2}} \ket{11} \Big).
$$ 

If we measure the first qubit, we have the same probability of obtaining 0 or 1. However, as the measurement collapses the state of the qubit, the two remaining qubits will be forced to be in a state that is somewhat linked to the one we have obtained for the first qubit (i.e., the part that is tensored with the result we obtain for the first qubit). Let us suppose that by measuring the first qubit, we have obtained a 1. Then, our 3-qubit system has collapsed to 
$$
\ket{\psi}_3 = \ket{1} \otimes \Big(\dfrac{1}{\sqrt{2}}\ket{01} + \dfrac{1}{\sqrt{2}}\ket{11}\Big),$$ which can also be seen as $$\ket{\psi}_3 = \ket{1} \otimes \Big(\dfrac{1}{\sqrt{2}}\ket{0} + \dfrac{1}{\sqrt{2}}\ket{1}\Big) \otimes \ket{1}.
$$ 

Note that the third qubit is already in one of the states of the computational basis, which means that, if we measure it right now, we will certainly obtain the value 1. The only remaining qubit that is not in the computational basis is the second one. Looking at the current state of our system, it is easy to see that we have the same probability of obtaining 0 or 1 by measuring it, which means that we will obtain 5 or 7.

\paragraph{Remark:} It is of interest to see if the result we obtain from one of the qubits will condition the possible values for the remaining qubits.  Let 
$$
\ket{\psi}_2 = \dfrac{1}{\sqrt{2}} (\ket{0} \otimes \ket{0}) + \dfrac{1}{\sqrt{2}} (\ket{1} \otimes \ket{1})
$$ 
be one of the four possible so-called Bell states \cite{bell1964einstein,nielsen2002quantum}, named after John Stewart Bell. This state is composed of two entangled qubits (they cannot be described as two single-qubit states). If we measure the second qubit, we have the same probability of obtaining 0 or 1. However, if we first measure the first qubit, and obtain 1, then the state of the second qubit will collapse (without having observed it) to 1, as the value 1 for the second qubit is only tensored with the value 1 of the first qubit. Thus, the result we obtain from a qubit or a set of qubits can be conditioned by the order in which we proceed to measure the rest of qubits. As will be seen later, the order in which we choose to read the members of a quantum register is one of the most important aspects of a quantum algorithm. \\
 
Another set of operations between qubits that are of great value are the inner and outer products, which we proceed to define.

\begin{defi}
Let $\ket{\psi_0}_n$ and $\ket{\psi_1}_n$ be two $n$-qubit systems, the \textbf{inner product} of $\ket{\psi_0}_n$ and $\ket{\psi_1}_n$ is the usual scalar product, defined by 
$$
\braket{\psi_0|\psi_1} = \ket{\psi_0}_n^* \ket{\psi_1}_n.
$$
\end{defi}

The inner product has the following well-known properties:

\begin{itemize}
\item $\braket{\psi_0|\psi_1} = \braket{\psi_1|\psi_0}^*$
\item $\braket{\psi_0|\ (a\ket{\psi_1} + b\ket{\psi_2})\ } = a \braket{\psi_0 | \psi_1} + b \braket{\psi_0 | \psi_2}$
\item $\braket{\psi | \psi} = || \ket{\psi}_n ||^2$
\end{itemize}

\paragraph{Remark:} As we are only considering unit vectors, $\braket{\psi|\psi} = 1$ for any quantum state $\ket{\psi}$.

\begin{defi}
Let $\ket{\psi_0}_n$ and $\ket{\psi_1}_n$ be two $n$-qubit systems, the \textbf{outer product} of $\ket{\psi_0}_n$ and $\ket{\psi_1}_n$ is defined by 
$$
\ket{\psi_0} \bra{\psi_1} = \ket{\psi_0}_n \ket{\psi_1}_n^*.
$$
\end{defi}

For example, let 
$$ 
\ket{\psi_0} = \alpha\ket{0} + \beta\ket{1} = \left[\begin{array}{ r }
\alpha \\
\beta
\end{array}\right] \mbox{ and }
\ket{\psi_1} = \gamma\ket{0} + \delta\ket{1} = \left[\begin{array}{ r }
\gamma \\
\delta
\end{array}\right]
$$ 
be two generic single-qubit systems, then the matrix representations of the inner and outer product between $\ket{\psi_0}$ and $\ket{\psi_1}$ are calculated as follows:
$$
\braket{\psi_0|\psi_1} = 
\left[\begin{array}{ r r} \alpha^* & \beta^* \end{array}\right] 
\left[\begin{array}{ r } \gamma \\ \delta \end{array}\right]
= \alpha^* \gamma + \beta^* \delta, \quad \quad
\ket{\psi_0} \bra{\psi_1} = 
\left[\begin{array}{ r } \alpha \\ \beta \end{array}\right] 
\left[\begin{array}{ r r } \gamma^* & \delta^* \end{array}\right] =
\left[\begin{array}{ r r } \alpha \gamma^* & \alpha \delta^* \\
\beta \gamma^* & \beta \delta^*\end{array}\right].
$$

\newpage

\section{Quantum Circuits}\label{sec:quantumcircuits}

The language of quantum circuits is a model of computation which is equivalent to quantum Turing machines or to universal quantum computers \cite{yao1993quantum}. Currently, it is the more extensively used when it comes to describe an algorithm that runs on a quantum machine, and draws upon a sequence of register measurements (as described in Section \ref{sec:qubits}) and discrete transformations (which will be explained in this section). This is mainly due because all its elements can be treated as classical, with the sole exception of the information that is going through the wires.  \\

First, we shall see what kind of transformations can be applied to the state of an $n$-qubit system. As a quantum state is always represented by a unitary vector, we need the most general operator that preserves this property and the dimension of the vector.

\begin{defi}
A matrix $A \in \mathcal{M}_{\mathbb{C}}(n)$ is unitary if
$$A^*A = AA^* = I$$
where $I$ is the identity matrix and $A^*$ is the Hermitian adjoint of $A$.
\end{defi}

\begin{defi}
The unitary group of degree $n$, denoted by $U_{\mathbb{C}}(n)$, is the group of $n \times n$ unitary matrices, with matrix multiplication as the group operation.
\end{defi}

\begin{prop}
Let $A \in U_{\mathbb{C}}(n)$ be a unitary matrix and let $x \in \mathbb{C}^n$ be a unitary vector, then $Ax \in \mathbb{C}^n$ is also a unitary vector.
\end{prop}

In this context, a unitary transformation acting on $n$-qubits is called an $n$-qubit quantum gate, and can be represented by a unitary matrix. Let us expand this concept and its physical implications (beyond the fact that {\em quantum gate} is probably cooler as a name than {\em unitary matrix}).

\begin{defi}
A quantum gate that operates on a space of one qubit is a linear operator represented (w.r.t. an orthogonal basis) by a  unitary matrix $A \in U_{\mathbb{C}}(2)$. More generally, a quantum gate acting on an $n$-qubit system is a linear operator, represented by a unitary matrix $A \in U_{\mathbb{C}}(2^n)$.
\end{defi}

Please note that quantum gates necessarily have the same number of inputs and outputs, as opposed to classical logic gates. From now on, all quantum gates will be represented with a bold symbol (and square brackets), in order to distinguish them from mere matrices. We will show the most frequently used quantum gates as examples, and various results that simplify in a dramatic way the difficulty of implementing physically any quantum gate.

\begin{defi}
The Hadamard gate is a single-qubit gate with the following matrix representation:
$$
\boldsymbol{H} = \frac{1}{\sqrt{2}} \left[\begin{array}{r r} 1 & 1 \\ 1 & -1 \end{array}\right]
$$
which is unitary.
\end{defi}

\paragraph{Remark:} The Hadamard gate, applied to each of the basis states, it has the following effect:
$$
\boldsymbol{H} : \ket{j} \rightarrow \frac{1}{\sqrt{2}} \Big(\ket{0} + (-1)^j\ket{1}\Big).
$$

\paragraph{Remark:} The classical flowchart for algorithms is replaced in quantum computing by the circuit representation, which allows one to get a glimpse of the procedure rather quickly. For instance, the circuit representation of the Hadamard gate is
\[
\Qcircuit @C=1em @R=.7em { \lstick{\ket{\psi_0}} & \gate{\boldsymbol{H}} & \rstick{\ket{\psi_1}} \qw }
\]
which is shorthand for $\Ket{\psi_1} \leftarrow \boldsymbol{H} \ket{\psi_0}$. The measurement step is written as
\[
\Qcircuit @C=1em @R=.7em { \lstick{\ket{\eta_0}} & \meter & \rstick{\ket{\eta_1}} \qw }
\]

\paragraph{Remark:}As will be seen when describing the most relevant quantum algorithms, the Hadamard transformation is one of the most important quantum gates. Its importance settles in the role it has at the time of generating all possible basis states, all of them with the same amplitude, inside a quantum register. 

Let us suppose that we have a qubit whose quantum state is $\ket{\psi_0} = \alpha \ket{0} + \beta \ket{1}$, where $\abs{\alpha}^2 + \abs{\beta}^2 = 1$. Then, the state of this single-qubit system after applying the Hadamard gate to it is:
$$
\boldsymbol{H} \ket{\psi_0} = \frac{1}{\sqrt{2}} \left[\begin{array}{r r} 1 & 1 \\ 1 & -1 \end{array}\right] \left[\begin{array}{r} \alpha \\ \beta \end{array}\right] = \frac{1}{\sqrt{2}} \left[\begin{array}{r} \alpha + \beta \\ \alpha - \beta \end{array}\right] = \frac{\alpha + \beta}{\sqrt{2}} \ket{0} + \frac{\alpha - \beta}{\sqrt{2}} \ket{1}
$$

Let us suppose now that, rather than having a generic state, we have the basis state $\ket{\psi_0} = \ket{0}$ in our one-qubit system. In that case, the result of applying the Hadamard gate will be as follows:
$$
\boldsymbol{H} \ket{\psi_0} = \frac{1}{\sqrt{2}} \left[\begin{array}{r r} 1 & 1 \\ 1 & -1 \end{array}\right] \left[\begin{array}{r} 1 \\ 0 \end{array}\right] = \frac{1}{\sqrt{2}} \left[\begin{array}{r} 1 \\ 1 \end{array}\right] = \frac{1}{\sqrt{2}} \ket{0} + \frac{1}{\sqrt{2}} \ket{1}
$$

As can be appreciated, we have transformed a basis state, $\ket{0}$, into a linear combination of the two basis states, $\ket{0}$ y $\ket{1}$, with identical amplitudes. If we measure the qubit at this moment, then we will obtain with equal probability one of the two possible basis states. That said, what will happen if, rather than having a single quantum state, we have a $n$-qubit quantum system, all of them also in their basis state $\ket{0}$?
$$
\boldsymbol{H}^{\otimes n} \ket{\psi_0}_n = \boldsymbol{H}^{\otimes n} \ket{0}_n = \frac{1}{\sqrt{2^n}} \left[\begin{array}{r r} 1 & 1 \\ 1 & -1 \end{array}\right]^{\otimes n}\left[\begin{array}{r} 1 \\ 0 \end{array}\right]^{\otimes n} = \frac{1}{\sqrt{2^n}} \left[\begin{array}{r} 1 \\ 1 \end{array}\right]^{\otimes n} = \frac{1}{\sqrt{2^n}} \sum_{j=0}^{2^n-1} \ket{j}_n
$$

What we have obtained is, thanks to quantum entanglement, a superposition of all basis states of the system with identical probability. In other words, if we measure our $n$-qubit register right now, we will obtain a certain integer $j \in \{ 0, \ldots , 2^n - 1 \}$ with probability $1/2^n$.\\

\begin{defi}
The Pauli gates are single-qubit gates with the following matrix representation:
$$
\boldsymbol{X} = \left[\begin{array}{r r} 0 & 1 \\ 1 & 0 \end{array}\right], \quad
\boldsymbol{Y} = \left[\begin{array}{r r} 0 & -i \\ i & 0 \end{array}\right], \quad 
\boldsymbol{Z} = \left[\begin{array}{r r} 1 & 0 \\ 0 & -1 \end{array}\right]
$$
which are unitary, but also Hermitian.
\end{defi}

The Pauli matrices have the following effect on the basis states:
\begin{eqnarray*}
\boldsymbol{X} : \ket{j} & \longmapsto & \ket{1 \oplus j} \\
\boldsymbol{Y} : \ket{j} & \longmapsto & (-i)^j \ket{1 \oplus j}\\
\boldsymbol{Z} : \ket{j} & \longmapsto & (-1)^j \ket{j}
\end{eqnarray*}

The three previous quantum gates are named after Wolfgang Pauli. The three of them, along with the identity matrix $I_2$, form a basis for the vector space of $2 \times 2$ Hermitian matrices (multiplied by real coefficients). \\

Single-qubit gates (that is, unitary $2 \times 2$ matrices) can be in fact, fully described by means of the following result, which is essentially straightforward:

\begin{prop}
Let $A \in U_{\mathbb C}(2)$. Then, there exist real numbers $\alpha$, $\beta$, $\gamma$ and $\delta$ such that
$$
A = e^{i\alpha} \left( \begin{array}{cc} e^{-i\beta/2} \\ \\ & e^{i\beta/2}  \end{array} \right) 
\left( \begin{array}{rr} \cos (\gamma/2) & -\sin (\gamma/2) \\ \\ \sin (\gamma/2) & \cos (\gamma/2)  \end{array} \right) 
\left( \begin{array}{cc} e^{-i\delta/2} \\ \\ & e^{-i\delta/2} \end{array} \right) 
$$
\end{prop}

\paragraph{Remark:} Matrices of the type
$$
\left( \begin{array}{cc} e^{-i\beta/2} \\ \\ & e^{i\beta/2}  \end{array} \right) 
$$
are usually called $z$--rotations, as their effect on a qubit, seen in the Bloch sphere, correspond precisely to a rotation of angle $\beta$ about the $z$ axis. For analogous reasons, matrices of the type
$$
\left( \begin{array}{rr} \cos (\gamma/2) & -\sin (\gamma/2) \\ \\ \sin (\gamma/2) & \cos (\gamma/2)  \end{array} \right) 
$$
are called $y$--rotations.\\

However, all previously defined quantum gates have their limitations. In fact, quantum gates that are the direct product of single-qubit gates cannot produce entanglement, which is reasonable, as entaglement needs at the very least two qubits in order to happen.

\begin{defi}
The $\boldsymbol{C}_{\NOT}$ gate, which stands for controlled-not, is a two-qubit quantum gate with the following matrix representation:
$$
\boldsymbol{C}_{\NOT} = \begin{bmatrix} 1 & 0 & 0 & 0 \\ 0 & 1 & 0 & 0 \\ 0 & 0 & 0 & 1 \\  0 & 0 & 1 & 0 \end{bmatrix}
$$
\end{defi}

Applied to a two-qubit basis state, the $\boldsymbol{C}_{\NOT}$ gate has the next effect: 
$$
\boldsymbol{C}_{\NOT} : \ket{i} \otimes \ket{j} \longmapsto \ket{i} \otimes \ket{i \oplus j }.
$$ 

The $\boldsymbol{C}_{\NOT}$ gate is another one of the key quantum gates, as it can be used to entangle and disentangle Bell states. In fact, it is the most simple gate that produces quantum entanglement. For example, let 
$$
\ket{\psi}_2 = \dfrac{1}{\sqrt{2}} \left( \ket{0}_2 + \ket{2}_2 \right)
$$ 
be a separable quantum state (as it can be written $\ket{\psi}_2 = \ket{+} \otimes \ket{0}$). If we apply the $\boldsymbol{C}_{\NOT}$ gate to it, we obtain
$$
\ket{\psi'}_2 = \boldsymbol{C}_{\NOT} \left( \ket{\psi}_2 \right)  = \dfrac{1}{\sqrt{2}} \Big( \boldsymbol{C}_{\NOT} (\ket{0} \otimes \ket{0}) + \boldsymbol{C}_{\NOT} (\ket{1} \otimes \ket{0})  \Big) = \dfrac{1}{\sqrt{2}} \left( \ket{0}\otimes\ket{0} + \ket{1}\otimes\ket{1} \right),
$$
which is one of the entangled Bell states.

\begin{thm}
{\em \cite{divincenzo1995two, barenco1995elementary}} Let $\boldsymbol{A} \in \mathcal{M}_{\mathbb{C}}(2^n)$ be an $n$-qubit gate, then it can be expressed as a finite number of tensor products of single-bit gates $M_i \in \mathcal{M}_{\mathbb{C}}(2)$ and the two-qubit $\boldsymbol{C}_{\NOT}$ gate.
\end{thm}

The previous result imply that every unitary transformation on an $n$-qubit system can be implemented physically using only single-qubit gates and the $\boldsymbol{C}_{\NOT}$ gate. In the usual parlance, single-qubit gates and the $\boldsymbol{C}_{\NOT}$ gates form a set of \textit{universal gates}.\\

Thus, we have explained the main notions needed for the correct comprehension of the rest of this paper: a quantum circuit algorithm will consist in a set of transformations of two different types, observations (measurements) and unitary transformations (quantum gates), to an $n$-qubit register. It is time then to describe some of the milestones of quantum computing.

\newpage

\section{Introduction to Quantum Algorithms}

We present a chronological summary of the first quantum algorithms that were shown to be more efficient than their best known classical counterparts. Our objective is to define them in the context of the previous sections, while showcasing their main properties and proving their correctness. We also give worked-out examples for some of them. \\

But first, let us explain two concepts that will be common to many of the algorithms here presented. The first one is the hidden subgroup problem, which we proceed to define:

\begin{defi}
Let $G$ be a group, let $K \subseteq G$ be a subgroup of $G$ and let $g \in G$. The cosets of $K$ in $G$ with respect to $g$ are the orbits $gK$ and $Kg$, called respectively the left coset and the right coset. 

Let $G$ be a finitely generated group, let $X$ be a finite set, and let $f: G \rightarrow X$ be a function that is constant on the (say) left cosets of a certain subgroup $K \subseteq G$ and distinct for each one of the cosets. The hidden subgroup problem, or {\em HSP}, is the problem of determining a generating set for $K$, using $f$ as a black box. 
\end{defi} 

Obviously there is no difference if the map $f$ is constant on the right cosets, it is an HSP likewise.\\

As will be shown, the superior performance of those algorithms relies on the ability of quantum computers to solve the hidden subgroup problem for finite Abelian groups. All those HSP-related algorithms were developed independently by different people, but the first to notice a common factor between them and to find a generalization was Richard Jozsa \cite{jozsa2001quantum}. \\

The other common factor, closely related to the HSP, is the possibility of building a quantum gate that can code a certain function $f$ that is given as a black box (i.e., as a digital circuit). A proof of this fact that uses the properties of reversible computation can be found in \cite{nielsen2002quantum}, and gives us another important quantum gate.

\begin{defi}
Let $f: \{0,1\}^n \rightarrow \{0,1\}^m$ be a function, the oracle gate $\boldsymbol{O}_{f}$ is the unitary transformation that has the following effect on the basis states of a quantum system:
$$
\boldsymbol{O}_{f} : \ket{j}_n \otimes \ket{k}_m \longmapsto \ket{j}_n \otimes \ket{k \oplus f(j)}_m,
$$ 
where $\oplus$ is the bitwise exclusive disjunction operation.
\end{defi}

The first algorithms that shared those elements eventually evolved into Shor's factoring algorithm, probably the most celebrated of all quantum algorithms and the one that gave birth to another of the greatest achievements in quantum computation: the quantum Fourier transform. All those algorithms are capable of solving their respective problems in polynomial time. However, for some of them the inexistence of polynomial-time classical algorithms for those same problems has yet to be proven. \\

Another class of algorithms, which will be shown at the end of the chapter and that also uses the oracle gate, are based on Grover's quantum search, whose objective is to speed up the finding of a solution for a problem whose candidate solutions can be verified in polynomial time (i.e., all problems in NP). \\

Finally, we explain the algorithm of quantum counting, which makes use of both worlds.

\newpage

\section{Deutsch's Algorithm}

Let $f: \{ 0,1 \} \rightarrow \{0,1\}$ be a function, it is clear that either $f(0) = f(1)$ or $f(0) \neq f(1)$. Let us suppose that we are given $f$ as a black box, and that we want to know if $f$ is constant. From a classical perspective, it is completely neccesary to evaluate the function both in $f(0)$ and $f(1)$ if we are to know this property with accuracy. Deutsch's algorithm \cite{deutsch1985quantum} shows us that, with the help of a quantum computer, it is possible to achieve this with only a query to $\boldsymbol{O}_f$. \\

We can see the previous question as an instance of the hidden subgroup problem, where $G = (\{0,1\},\oplus)$, $X = \{0,1\}$, and $K$ is either $\{0\}$ or $\{0,1\}$ depending on the nature of $f$. Note that in this case the cosets of $\{0\}$ are $\{0\}$ and $\{1\}$ and that the only coset of $\{0,1\}$ is precisely $\{0,1\}$.

\paragraph{$\mathbb{SETUP}$}
$$$$\noindent\framebox{\parbox[b]{\linewidth}{\begin{algorithmic}
\State $\Ket{\psi_0}_{1,1} \leftarrow \Ket{0} \otimes \Ket{1}$
\end{algorithmic}}} \\

Deutsch's algorithm needs only two one-qubit registers. The first one is initialized at $\Ket{0}$, and the second one at $\Ket{1}$. As will be seen, this is due to the properties of the Hadamard gate when applied to the canonical basis states, and it is a frequent way of initializing a quantum algorithm.

\paragraph{$\mathbb{STEP}$ 1}
$$$$\noindent\framebox{\parbox[b]{\linewidth}{\begin{algorithmic}
\State $\Ket{\psi_1}_{1,1} \leftarrow \boldsymbol{H}^{\otimes 2} \left( \ket{\psi_0}_{1,1} \right)$
\end{algorithmic}}} \\

On the first step of Deutsch's algorithm we apply the Hadamard gate to both quantum registers, thus transforming the values of the canonical basis into the respective ones of the Hadamard basis. \\
$$
\ket{\psi_1}_{1,1} = \boldsymbol{H}^{\otimes 2} \left( \ket{0} \otimes \ket{1} \right) = \left( \boldsymbol{H} \ket{0} \right)  \otimes  \left( \boldsymbol{H} \ket{1} \right) = \left( \dfrac{\ket{0} + \ket{1}}{\sqrt{2}}\right) \otimes  \left(\dfrac{\ket{0} - \ket{1}}{\sqrt{2}}\right) = \ket{+} \otimes \ket{-}
$$

\paragraph{$\mathbb{STEP}$ 2}
$$$$\noindent\framebox{\parbox[b]{\linewidth}{\begin{algorithmic}
\State $\Ket{\psi_2}_{1,1} \leftarrow \boldsymbol{O}_{f} \left( \ket{\psi_1}_{1,1} \right)$
\end{algorithmic}}} \\

The second step needs the oracle gate, defined in the previous chapter for a generic function and for $n$ and $m$ qubits. Note that, in this case, the function $f$ associated to the oracle gate as a black box is the one given for this instance of the hidden subgroup problem: $f: \{ 0,1 \} \rightarrow \{0,1\}$. In fact, the oracle gate is not a constant transformation as are the Hadamard  or the Pauli gates, but it rather depends on the problem. It is thus constructed \textit{ad hoc} subject to the question we want to answer, provided that we have a logic circuit that implements $f$. The effect the oracle gate has on our quantum register can be seen as follows:
\begin{eqnarray*}
\ket{\psi_2}_{1,1} & =  & \boldsymbol{O}_{f} \left( \ket{\psi_1}_{1,1} \right) = \boldsymbol{O}_{f} \left( \ket{+} \otimes \ket{-} \right)  = \boldsymbol{O}_{f} \Big[ \Big(\dfrac{\ket{0} + \ket{1}}{\sqrt{2}}\Big) \otimes  \ket{-} \Big] \\ 
& = & \dfrac{\boldsymbol{O}_{f}(\ket{0} \otimes \ket{-}) + \boldsymbol{O}_{f}(\ket{1}\otimes\ket{-})}{\sqrt{2}} 
= \dfrac{(-1)^{f(0)} \ket{0} \otimes \ket{-} + (-1)^{f(1)} \ket{1} \otimes \ket{-}}{\sqrt{2}} \\ 
& = & \left( \dfrac{(-1)^{f(0)} \ket{0} + (-1)^{f(1)} \ket{1}}{\sqrt{2}} \right) \otimes \ket{-}
\end{eqnarray*}

Please notice that all operations are just algebraic manipulations which allow us to see more clearly the information we have inside our quantum computer. We are not modifying anything, we are just reshaping the equation in order to have a better picture of what is happening.

\paragraph{$\mathbb{STEP}$ 3}
$$$$\noindent\framebox{\parbox[b]{\linewidth}{\begin{algorithmic}
\State $\Ket{\psi_3}_{1,1} \leftarrow \left( \boldsymbol{H} \otimes \boldsymbol{I} \right) \left( \ket{\psi_2}_{1,1} \right)$
\end{algorithmic}}} \\

The third and final step before measuring our quantum register involves again the Hadamard gate $\boldsymbol{H}$, but this time it is only applied to the first register. The second register is left alone, which is represented with an identity gate $\boldsymbol{I}$. In fact, the information inside the second register is no longer relevant, as it was only used as the auxiliary register needed for the oracle gate.

\begin{eqnarray*}
\ket{\psi_3}_{1,1} & = &\left( \boldsymbol{H} \otimes \boldsymbol{I} \right) \left( \ket{\psi_2}_{1,1} \right) = \left( \boldsymbol{H} \otimes \boldsymbol{I} \right) \left[ \left( \dfrac{(-1)^{f(0)} \ket{0} + (-1)^{f(1)} \ket{1}}{\sqrt{2}} \right) \otimes \ket{-} \right] \\
& = & \left( \dfrac{(-1)^{f(0)} \boldsymbol{H} \ket{0} + (-1)^{f(1)} \boldsymbol{H} \ket{1}}{\sqrt{2}} \right) \otimes \ket{-} = \left( \dfrac{(-1)^{f(0)} \ket{+} + (-1)^{f(1)} \ket{-}}{\sqrt{2}} \right) \otimes \ket{-} \\
& = & \left( \dfrac{(-1)^{f(0)} \ket{0} + (-1)^{f(0)} \ket{1} + (-1)^{f(1)} \ket{0} - (-1)^{f(1)} \ket{1}}{2} \right) \otimes \ket{-} \\
& = &\left( \dfrac{[(-1)^{f(0)} + (-1)^{f(1)}]\ket{0} + [(-1)^{f(0)} - (-1)^{f(1)}]\ket{1}}{2} \right) \otimes \ket{-} \\
& = & (-1)^{f(0)} \ket{f(0) \oplus f(1)} \otimes \ket{-}
\end{eqnarray*}

In order to understand the last part of the equation, one must take into account that, if $f(0) = f(1)$, then $(-1)^{f(0)} - (-1)^{f(1)} = 0$ and $f(0) \oplus f(1) = 0$. A similar reasoning goes for $f(0) \neq f(1)$, which leads us to the last expression for $\ket{\psi_3}_{1,1}$. Note also that $f(0) \oplus f(1) = 0$ if and only if $f(0) = f(1)$, and that $f(0) \oplus f(1) = 1$ if and only if $f(0) \neq f(1)$.

\paragraph{$\mathbb{STEP}$ 4}
$$$$\noindent\framebox{\parbox[b]{\linewidth}{\begin{algorithmic}
\State $\tilde \delta \leftarrow$ measure the first register of $\Ket{\psi_3}_{1,1}$
\end{algorithmic}}} \\

As the reader has surely noted, the information we wanted to obtain from the function $f$ is already in the first register. We measure it now, thus destroying all the information related to the amplitudes of the basis states, and obtain a certain $\tilde \delta \in \{0,1\}$. If $\tilde \delta = 0$, then $K = \{0,1\}$ and $f(0) = f(1)$. If $\tilde \delta = 1$, then $K = \{0\}$ and $f(0) \neq f(1)$. A circuit representation of Deutsch's algorithm can be found below. \\

\begin{figure}\label{fig::deutschalgorithm}
\caption{Circuit representation of Deutsch's algorithm} 
\[
\Qcircuit @C=2em @R=2em {
  \lstick{\ket{0}} & \qw & \gate{\boldsymbol{H}} & \qw & \multigate{1}{\boldsymbol{O}_{f}} & \qw &  \gate{\boldsymbol{H}} & \qw & \meter & \cw & \tilde{\delta} \\
  \lstick{\ket{1}} & \qw & \gate{\boldsymbol{H}} & \qw & \ghost{\boldsymbol{O}_{f}} & \qw & \qw & \qw & \qw & \qw \\
  & \ket{\psi_0} & & \ket{\psi_1} & & \ket{\psi_2} & & \ket{\psi_3} & & }
\]
\end{figure}
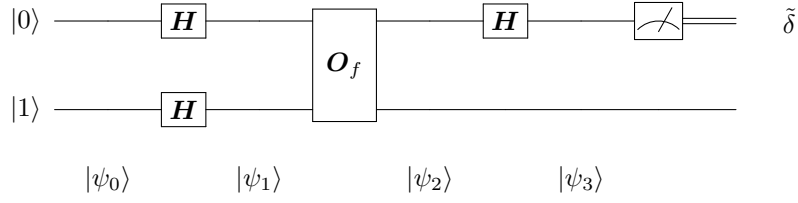

At this moment, the inherent capabilities of quantum computing begin to surface. A problem which needs two evaluations of a function $f$ in its classical version, can be reduced to just one evaluation of the same function in its quantum counterpart thanks to quantum parallelism. One may wonder if this property could be scaled to a function acting on $\{0,1\}^n$ rather than just $\{0,1\}$. That is the objective of the next algorithm.

\newpage

\section{Deutsch–Jozsa Algorithm}

The following algorithm is a generalization of the previous one. Its original version appeared in \cite{deutsch1992rapid}, and is due again to David Deutsch and also to Richard Jozsa. Let $f: \{ 0,1 \}^n \rightarrow \{0,1\}$ be a function that is either constant for all values in $\{ 0,1 \}^n$, or is else balanced (i.e., equal to 0 for exactly half of all possible values in $\{ 0,1 \}^n$, and to 1 for the other half). The problem of determining if the function $f$ is constant or balanced, using it as a black box, is called Deutsch's problem. In the classical version, a solution for this problem requires $2^{n-1}+1$ evaluations of $f$ in the worst case. Let us see if we can improve that bound with the help of a quantum computer. 

\paragraph{$\mathbb{SETUP}$}
$$$$\noindent\framebox{\parbox[b]{\linewidth}{\begin{algorithmic}
\State $\Ket{\psi_0}_{n,1} \leftarrow \Ket{0}_n \otimes \Ket{1}$
\end{algorithmic}}} \\

We need a quantum computer with $n+1$ qubits, where the first $n$ qubits will be initialized at $\ket{0}$ and the remaining one at $\ket{1}$. Again, the single-qubit register is only used as the auxiliary qubit required for the oracle gate.

\paragraph{$\mathbb{STEP}$ 1}
$$$$\noindent\framebox{\parbox[b]{\linewidth}{\begin{algorithmic}
\State $\Ket{\psi_1}_{n,1} \leftarrow \boldsymbol{H}^{\otimes n+1} \left( \ket{\psi_0}_{n,1} \right)$
\end{algorithmic}}} \\

The first transformation we apply to our system is again the Hadamard gate. As explained before, when applied to the basis state $\ket{0}_n$ the Hadamard transformation gives us a superposition of all basis states with identical probability, thus obtaining the following quantum state:
$$
\ket{\psi_1}_{n,1} = \left( \boldsymbol{H}^{\otimes n} \ket{0}_n \right) \otimes  \left( \boldsymbol{H} \ket{1} \right) = \left( \dfrac{1}{\sqrt{2^n}} \sum_{i=0}^{2^n-1} \ket{i}_n \right) \otimes \ket{-}
$$

\paragraph{$\mathbb{STEP}$ 2}
$$$$\noindent\framebox{\parbox[b]{\linewidth}{\begin{algorithmic}
\State $\Ket{\psi_2}_{n,1} \leftarrow \boldsymbol{O}_{f} \left( \Ket{\psi_1}_{n,1} \right)$
\end{algorithmic}}} \\

Now, we apply the oracle gate, which in this case is constructed for $n+1$ qubits  and for the specific function $f$ that we want to know if it is constant or balanced.
$$
\Ket{\psi_2}_{n,1} = \boldsymbol{O}_{f} \Big[ \Big( \dfrac{1}{\sqrt{2^n}} \sum_{i=0}^{2^n-1} \ket{i}_n \Big) \otimes \ket{-} \Big] = \dfrac{1}{\sqrt{2^n}} \sum_{j=0}^{2^n -1} \boldsymbol{O}_{f} ( \ket{j}_n \otimes \ket{-} ) = \dfrac{1}{\sqrt{2^n}} \sum_{j=0}^{2^n -1} (-1)^{f(j)} \ket{j}_n \otimes \ket{-}
$$

The last step is better understood if we apply it separately to a generic basis state $\ket{j}_n$ with $j \in \{0,\ldots,2^n-1\}$, tensored with the Hadamard basis state $\ket{-}$.
\begin{eqnarray*}
\boldsymbol{O}_{f} \Big( \ket{j}_n \otimes \ket{-} \Big) & = & \boldsymbol{O}_{f} \Big( \ket{j}_n \otimes \dfrac{\ket{0} - \ket{1}}{\sqrt{2}} \Big) = \dfrac{\boldsymbol{O}_{f} (\ket{j}_n \otimes \ket{0}) - \boldsymbol{O}_{f} (\ket{j}_n \otimes \ket{1})}{\sqrt{2}} \\
& = & \dfrac{\ket{j}_n \otimes \ket{f(j)} - \ket{j}_n \otimes \ket{1 \oplus f(j)} }{\sqrt{2}} = (-1)^{f(j)} \ket{j}_n \otimes \Big( \dfrac{\ket{0} - \ket{1}}{\sqrt{2}} \Big) \\
& = & (-1)^{f(j)} \ket{j}_n \otimes \ket{-}
\end{eqnarray*}

Thus, we end up again with a superposition of all basis states in the first register, all of them with identical probability. The only difference with the previous state is that the amplitude of the states $\ket{j}_n$ remains positive if $f(j)=0$, and becomes negative when $f(j)=1$. Taking into account that either $f$ is constant or balanced (i.e., all amplitudes are the same now, or half the amplitudes are positive and the other ones negative), is there any way to obtain this information from our quantum register? Note that, until now, although we have applied $f$ to every possible $j \in \{0,\ldots,2^n-1\}$, we have used the gate that implements it only once.

\paragraph{$\mathbb{STEP}$ 3}
$$$$\noindent\framebox{\parbox[b]{\linewidth}{\begin{algorithmic}
\State $\Ket{\psi_3}_{n,1} \leftarrow ( \boldsymbol{H}^{\otimes n} \otimes \boldsymbol{I}) \left( \Ket{\psi_2}_{n,1} \right)$
\end{algorithmic}}} \\

The last step involves again the Hadamard transform. We apply it to the first $n$ qubits of our quantum system, and obtain the following:
\begin{eqnarray*}
\Ket{\psi_3}_{n,1} & = & (\boldsymbol{H}^{\otimes n} \otimes \boldsymbol{I}) \left( \dfrac{1}{\sqrt{2^n}} \sum_{j=0}^{2^n -1} (-1)^{f(j)} \ket{j}_n \otimes \ket{-} \right) = \left( \dfrac{1}{\sqrt{2^n}} \sum_{j=0}^{2^n -1} (-1)^{f(j)} \boldsymbol{H}^{\otimes n} \ket{j}_n \right) \otimes \ket{-} \\
 & = & \left[ \dfrac{1}{2^n} \sum_{k=0}^{2^n-1} \left( \sum_{j=0}^{2^n -1} (-1)^{f(j) + j \cdot k} \right) \ket{k}_n \right] \otimes \ket{-}
\end{eqnarray*}

In order to understand the last equation, we must first fathom the effects of the Hadamard gate on an $n$-qubit basis state. Let $j \in \{0,\ldots,2^n-1\}$, a closer inspection leads us to the following identity:
$$
\boldsymbol{H}^{\otimes n} \left( \ket{j}_n \right)  = \frac{1}{\sqrt{2^n}} \sum_{k=0}^{2^n-1} (-1)^{j \cdot k} \ket{k}_n,
$$
where $j \cdot k$ is the bitwise inner product of $j$ and $k$, modulo 2, i.e., 
$$
j \cdot k = \sum_{l=1}^{n} j_l k_l \mod 2,
$$ 
with $j_l$ and $k_l$ being the binary digits of $j$ and $k$ respectively. Likewise, the last identity is a generalization of the effect of the one-qubit Hadamard gate, which can be seen as:
$$
\boldsymbol{H} ( \ket{j} ) = \frac{1}{\sqrt{2}} \Big(\ket{0} + (-1)^j\ket{1}\Big) = \frac{1}{\sqrt{2}} \sum_{k=0}^{1} (-1)^{jk} \ket{k}
$$

\paragraph{$\mathbb{STEP}$ 4}
$$$$\noindent\framebox{\parbox[b]{\linewidth}{\begin{algorithmic}
	\State $\tilde k \leftarrow$ measure the first register of $\Ket{\psi_3}_{n,1}$
\end{algorithmic}}} \\

Finally, we are able to measure the qubits, thus destroying the information inside the register and obtaining a number $\tilde k \in \{0,\ldots,2^n-1\}$ according to the probability distribution given by the amplitudes $$\alpha_k = \sum_{j=0}^{2^n -1} (-1)^{f(j) + j \cdot k}.$$

Let us have a closer look to the probability of obtaining $\tilde k = 0$. It is given by:
$$
|\alpha_0|^2 = \left| \dfrac{1}{\sqrt{2^n}} \sum_{j=0}^{2^n -1} (-1)^{f(j)} \right|^2.
$$

It is easy to see that $|\alpha_0|^2 = 1$ if and only if the function $f$ is constant for all $j \in \{0,\ldots,2^n-1\}$, and that $|\alpha_0|^2 = 0$ if and only if the function $f$ is balanced (recall that we are promised that $f$ is of one of those two natures). Having said that, if we measure now the first register and obtain some $\tilde k$, we can conclude that $f$ is constant if $\tilde k = 0$, and that $f$ is balanced otherwise. With just a single evaluation of $f$ we have answered the question, as opposed to the $2^{n-1} + 1$ evaluations needed in the classical version. A circuit representation of Deutsch-Jozsa algorithm is displayed in Figure 4. \\

The previous algorithm has much more profound implications than the possibility of solving Deutsch's problem exponentially faster with the help of a quantum computer. It also tells us that, relative to an oracle (i.e., a black box that solves a certain problem or function, namely $f$) we can establish a difference between the classes P and EQP (Exact Quantum Polynomial, that is, quantum algorithms that run in polynomial time and give the solution with probability $1$). Note that this does not imply that P $\neq$ EQP, it just tells us that there exists an oracle separation between P and EQP. We shall come back to these concepts later on.

\begin{figure}\label{fig::djalgorithm}
\caption{Circuit representation of Deutsch-Jozsa algorithm}
\[
\Qcircuit @C=2em @R=2em {
  \lstick{\ket{0}_n} & \qw & \gate{\boldsymbol{H}^{\otimes n}} & \qw & \multigate{1}{\boldsymbol{O}_{f}} & \qw & \gate{\boldsymbol{H}^{\otimes n}} & \qw & \meter & \cw & \tilde{k} \\
  \lstick{\ket{1}} & \qw & \gate{\boldsymbol{H}} & \qw & \ghost{\boldsymbol{O}_{f}} & \qw & \qw & \qw  & \qw & \qw \\
  & \ket{\psi_0} & & \ket{\psi_1} & & \ket{\psi_2} & & \ket{\psi_3} & & }
\]
\end{figure}
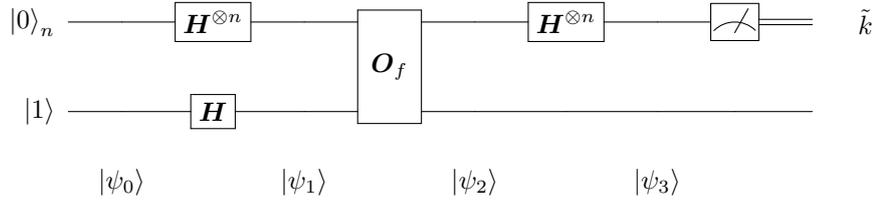

\newpage

\section{Simon's Algorithm}

Let $f: \{0,1\}^n \rightarrow \{0,1\}^n$ be a function such that, for some $s \in \{0,1\}^n$ with $s \neq 0$, $f(j) = f(k)$ if and only if either $j = k$ or $j \oplus k = s$ for all $j,k \in \{0,1\}^n$ (where $\oplus$ is again the bitwise exclusive disjunction operation, also called bitwise xor). Simon's problem is defined as: given such an $f$ as a black box, figure out the value of $s$, which is usually called the xor-mask of $f$. Both the problem and the quantum algorithm we proceed to explain were both first presented in \cite{simon1997power} by Daniel R. Simon, hence their names. \\

Simon's problem can also be seen as an instance of the hidden subgroup problem, where $G = (\{0,1\}^n,\oplus)$, $X \subseteq \{0,1\}^n$ is any finite set, and $K = \{0,s\}$ for some $s \in \{0,1\}^n$. In the classical version, a solution for this problem requires that we find a pair of values $x,y \in \{0,1\}^n$ such that  $f(x) = f(y)$, and then compute $x \oplus y$. This solution requires $\mathcal{O}(2^{n/2})$ evaluations of $f$ in the worst case whereas, as will be proved later, Simon's algorithm only needs $\mathcal{O}(n)$ evaluations of $f$.

\paragraph{$\mathbb{SETUP}$}
$$$$\noindent\framebox{\parbox[b]{\linewidth}{\begin{algorithmic}
\State $\Ket{\psi_0}_{n,n} \leftarrow \Ket{0}_n \otimes \Ket{0}_n$
\end{algorithmic}}} \\

In this algorithm, we need $2n$ qubits, all of them initialized at $\ket{0}$.

\paragraph{$\mathbb{STEP}$ 1}
$$$$\noindent\framebox{\parbox[b]{\linewidth}{\begin{algorithmic}
\State $\Ket{\psi_1}_{n,n} \leftarrow ( \boldsymbol{H}^{\otimes n} \otimes \boldsymbol{I}^{\otimes n} ) \left( \Ket{\psi_0}_{n,n} \right)$
\end{algorithmic}}} \\

We first apply the Hadamard transformation to the first half of our qubit set, thus obtaining the following quantum state.
$$
\ket{\psi_1}_{n,n} = ( \boldsymbol{H}^{\otimes n} \otimes \boldsymbol{I}^{\otimes n} ) (\Ket{0}_n \otimes \Ket{0}_n) = (\boldsymbol{H}^{\otimes n} \Ket{0}_n ) \otimes (\Ket{0}_n) = \dfrac{1}{\sqrt{2^n}} \sum_{j=0}^{2^n - 1}  \Ket{j}_n \otimes \Ket{0}_n 
$$

\paragraph{$\mathbb{STEP}$ 2}
$$$$\noindent\framebox{\parbox[b]{\linewidth}{\begin{algorithmic}
\State $\Ket{\psi_2}_{n,n} \leftarrow \boldsymbol{O}_{f} \left( \Ket{\psi_1}_{n,n} \right)$
\end{algorithmic}}} \\

Next, we use the oracle gate, built particularly for the function $f$. Note that, thanks to quantum parallelism, we apply here the function $f$ to all possible values in $\{0,1\}^n$ with just a single iteration of $\boldsymbol{O}_{f}$. Thus, all possible values of $f$ are now present in the second register.
$$
\Ket{\psi_2}_{n,n} = \boldsymbol{O}_{f} \left( \dfrac{1}{\sqrt{2^n}} \sum_{j=0}^{2^n - 1}  \Ket{j}_n \otimes \Ket{0}_n \right) = \dfrac{1}{\sqrt{2^n}} \sum_{j=0}^{2^n - 1} \boldsymbol{O}_{f} \left( \Ket{j}_n \otimes \Ket{0}_n \right) = \dfrac{1}{\sqrt{2^n}} \sum_{j=0}^{2^n - 1}  \Ket{j}_n \otimes \Ket{f(j)}_n 
$$

\paragraph{$\mathbb{STEP}$ 3}
$$$$\noindent\framebox{\parbox[b]{\linewidth}{\begin{algorithmic}
\State $\tilde \delta \leftarrow$ measure the second register of $\Ket{\psi_2}_{n,n}$
\State
\State $\Ket{\psi_3}_{n} \leftarrow \Ket{\psi_2}_{n,n}$ after measuring the second register 
\end{algorithmic}}} \\

In this step we see for the first time the true effects of measuring part of our quantum system before completing the execution of an algorithm. If we measure now the second register, we shall end up with a value $\tilde \delta = f(\tilde j)$ for a certain $\tilde j \in \{0,1\}^n$. Thus, only the values in $f^{-1}(\tilde \delta)$ will remain in the first register (before measuring, they were the only ones tensored with $\ket{\tilde \delta}$). In any case, as $f^{-1}(\tilde \delta) = \{\tilde j, \tilde j \oplus s\}$ (we remark that $j \oplus k = s$ if and only if $j \oplus s = k)$, we end up with the following quantum state:

$$\Ket{\psi_3}_{n} = \dfrac{1}{\sqrt{2}} \left( \ket{\tilde j}_n + \ket{\tilde j \oplus s}_n \right)$$

\paragraph{$\mathbb{STEP}$ 4}
$$$$\noindent\framebox{\parbox[b]{\linewidth}{\begin{algorithmic}
\State $\Ket{\psi_4}_{n} \leftarrow \boldsymbol{H}^{\otimes n}  \left( \Ket{\psi_3}_{n} \right)$
\end{algorithmic}}} \\

The last transformation we apply to our quantum system is, again, the Hadamard gate. Before that, we could have measured the first register and obtain a certain value in $f^{-1}(\tilde \delta) = \{\tilde j, \tilde j \oplus s\}$. However, in that case we would have ended with the same information as if we would have just made a single classical evaluation of $f$. The Hadamard gate, on the other hand, will let us obtain much more information than from a single evaluation of $f$: it will give us some precious information about $s$. Following a similar reasoning as with Deutsch-Jozsa algorithm, we end up with:
\begin{eqnarray*}
\Ket{\psi_4}_{n} & = & \boldsymbol{H}^{\otimes n} \left[ \dfrac{1}{\sqrt{2}} \left( \ket{\tilde j}_n + \ket{\tilde j \oplus s}_n \right) \right] = \dfrac{1}{\sqrt{2}} \left( \boldsymbol{H}^{\otimes n} \ket{\tilde j}_n + \boldsymbol{H}^{\otimes n} \ket{\tilde j \oplus s}_n \right) \\
 & = & \dfrac{1}{\sqrt{2}} \left[ \dfrac{1}{\sqrt{2^n}} \sum_{k=0}^{2^n-1} (-1)^{\tilde j \cdot k} \ket{k}_n + \dfrac{1}{\sqrt{2^n}} \sum_{k=0}^{2^n-1} (-1)^{(\tilde j \oplus s) \cdot k} \ket{k}_n \right] \\ 
 & = & \dfrac{1}{\sqrt{2^{n+1}}} \sum_{k=0}^{2^n-1} (-1)^{\tilde j \cdot k + (\tilde j \oplus s) \cdot k} \ket{k}_n = \dfrac{1}{\sqrt{2^{n+1}}} \sum_{k=0}^{2^n-1} (-1)^{\tilde j \cdot k} \left[ 1 + (-1)^{s \cdot k} \right] \ket{k}_n
\end{eqnarray*}

\paragraph{$\mathbb{STEP}$ 5}
$$$$\noindent\framebox{\parbox[b]{\linewidth}{\begin{algorithmic}
\State $\tilde \omega \leftarrow$ measure $\Ket{\psi_4}_{n}$
\end{algorithmic}}} \\

Let us suppose that we measure our quantum system right now. It is clear that the current amplitudes of the basis states are $$\alpha_k = \left| \dfrac{1}{\sqrt{2^{n+1}}} \left[ 1 + (-1)^{s \cdot k} \right] \right|^2$$ for $k \in \{0,1\}^n$. \\

However, it can be noted that $\alpha_k \neq 0$ if and only if $s \cdot k \equiv 0 \mod 2$, which happens for half the values of $k$ necessarily. Even more, in those cases the amplitude is equal to $1 / 2^{n-1}$. Analyzing the outcome, we have ended up with some $\tilde \omega$ such that $\tilde \omega \cdot s \equiv 0 \mod 2$. If we are able to find $n-1$ linearly independent values of $\tilde \omega$, namely $\tilde \omega_1,\ldots,\tilde \omega_{n-1}$, we will arrive at a system of equations whose solutions are $0$ and $s$. Before proving this, let us explain the performance of Simon's algorithm with a worked-out example.

\begin{figure}
\caption{Circuit representation of Simon's algorithm (one iteration)}
\[  
\Qcircuit @C=2em @R=2em {
  \lstick{\ket{0}_n} & \qw 
  & \gate{\boldsymbol{H}^{\otimes n}} & \qw
  & \multigate{1}{\boldsymbol{O}_{f}} & \qw 
  & \qw & \qw
  & \gate{\boldsymbol{H}^{\otimes n}} & \qw 
  & \meter & \cw & \tilde \omega \\
  \lstick{\ket{0}_n} & \qw 
  & \qw & \qw 
  & \ghost{\boldsymbol{O}_{f}} & \qw 
  & \meter & \cw 
  & \cw & \cw
  & \cw & \cw & \tilde \delta \\
  & \ket{\psi_0} & 
  & \ket{\psi_1} & 
  & \ket{\psi_2} &
  & \ket{\psi_3} &
  & \ket{\psi_4} &
  &  }
\]
\end{figure}
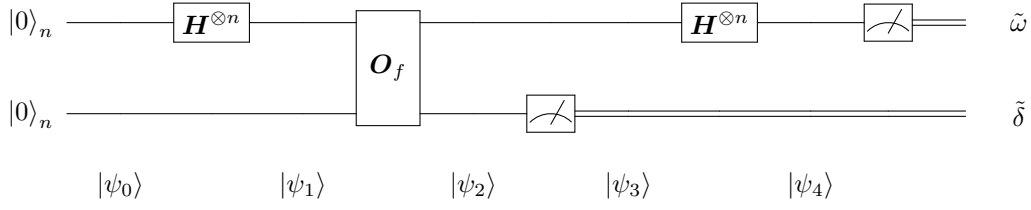

\subsection*{Example with $n=4$}

Let us suppose that we are given as a black box a function that fulfills the requirements of Simon's problem. This function, namely $f : \{0,1\}^4 \rightarrow \{0,1\}^4$, has the following outcome:

\begin{center}
\begin{tabular}{lclcl c lclcl}
$f(0)$ & $=$ & $f(5)$ & $=$ & $0$ & \quad \quad \quad \quad & $f(1)$ & $=$ & $f(4)$ & $=$ & $1$ \\
$f(2)$ & $=$ & $f(7)$ & $=$ & $2$ & & $f(3)$ & $=$ & $f(6)$ & $=$ & $3$ \\
$f(8)$ & $=$ & $f(13)$ & $=$ & $4$ & & $f(9)$ & $=$ & $f(12)$ & $=$ & $5$ \\
$f(10)$ & $=$ & $f(15)$ & $=$ & $6$ & & $f(11)$ & $=$ & $f(14)$ & $=$ & $7$ \\
\end{tabular}
\end{center}

Of course, as the function is given as a black box, this information is only available to us if we evaluate $f(j)$ for all $j \in \{0,\ldots,15\}$. A closer inspection of these values tells us that this function has in fact a xor-mask and that its value is $s=5$. Our objective is to arrive at this knowledge without evaluating $f$ classically for all values in $\{0,1\}^4$. \\

Now, let us suppose that we do not know this information yet. As a start, Simon's algorithm would need the quantum state $$\Ket{\psi_0}_{4,4} = \Ket{0}_4 \otimes \Ket{0}_4.$$After applying the Hadamard gate, we would obtain $$\Ket{\psi_1}_{4,4} = \dfrac{1}{16} \sum_{j=0}^{15}  \Ket{j}_4 \otimes \Ket{0}_4$$ and, after the oracle gate, our system is in the state $$\ket{\psi_2}_{4,4} = \dfrac{1}{16} \sum_{j=0}^{15} \ket{j}_4 \otimes \ket{f(j)}_4.$$

If we make use of the information we know (but we should not!) about $f$, we could see the previous equation as:
\begin{equation*}
\begin{split}
\Ket{\psi_2}_{4,4} = \dfrac{1}{8}\Big[ & \big(\ket{0}_4 + \ket{5}_4\big)\otimes \ket{0}_4 + \big(\ket{1}_4 + \ket{4}_4\big)\otimes \ket{1}_4 + \big(\ket{2}_4 + \ket{7}_4\big)\otimes \ket{2}_4 + \\
 & \big(\ket{3}_4 + \ket{6}_4\big)\otimes \ket{3}_4 + \big(\ket{8}_4 + \ket{13}_4\big)\otimes \ket{4}_4 + \big(\ket{9}_4 + \ket{12}_4\big)\otimes \ket{5}_4 + \\
 & \big(\ket{10}_4 + \ket{15}_4\big)\otimes \ket{6}_4 + \big(\ket{11}_4 + \ket{14}_4\big)\otimes \ket{7}_4 \Big]
\end{split}
\end{equation*}

Please note again that the previous state is actually happening inside our quantum computer whether or not we know the values for $f(j)$. As we have constructed our oracle gate using $f$ as a black box, it necessarily has the previous effect on the Hadamard state. \\

Up until now all steps were deterministic. However, the next step, the measurement of the second register, will have a non-deterministic outcome. Let us suppose that we measure it and obtain, for example, the value $\tilde \delta = 6$. Necessarily, our quantum system is now in the state $$\Ket{\psi_3}_{4} = \dfrac{1}{\sqrt{2}} (\ket{10}_4 + \ket{15}_4)$$ and, if we apply now the Hadamard transform, we end up with $$\ket{\psi_4}_4 = \dfrac{1}{\sqrt{2^{5}}} \sum_{k=0}^{15} (-1)^{\tilde j \cdot k} \left[ 1 + (-1)^{s \cdot k} \right] \ket{k}_n,$$ where $\tilde j \in f^{-1}(6) = \{10,15\}$ and $s$ is the (yet unknown!) xor-mask of $f$. \\

As can be seen, the values we get of $\tilde j$ and $\tilde \delta$ are unimportant. What we need is the outcome of the measurement of our quantum system at this moment. As previously said, we will end up with a number $\tilde \omega$ such that $s \cdot \tilde \omega = 0 \mod 2$. In this case, $\tilde \omega \in \{0,2,5,7,8,10,13\}$ for $s=5$, where all of them have the same probability of coming up. \\

Let us suppose that we have run three complete iterations of Simon's algorithm, and obtain $\tilde \omega_1 = 2$, $\tilde \omega_2 = 7$ and $\tilde \omega_3 = 10$. Thus, draining from the fact that $\tilde \omega_i \cdot s = 0 \mod 2$ for all of them, if we define $s = s_3 s_2 s_1 s_0$ as the bitwise representation of $s$ (with $s_i \in \{0,1\}$), we can consider the system of equations

\begin{center}
\begin{tabular}{ c c c c c c c c l}
 & & & & $s_1$ & & & $=$ & $0 \mod 2$ \\
 & & $s_2$ & $+$ & $s_1$ & $+$ & $s_0$ & $=$ & $0 \mod 2$ \\
$s_3$ & $+$ & $s_2$ & & & $+$ & $s_0$ & $=$ & $0 \mod 2,$ 
\end{tabular}
\end{center} with each one of the equations giving us respectively the following set of solutions:

\begin{center}
\begin{tabular}{ c c l}
$\Omega_1$ & $=$ & $\{0,1,4,5,8,9,12,13\},$ \\
$\Omega_2$ & $=$ & $\{0,3,5,6,8,11,13,14\},$ \\
$\Omega_3$ & $=$ & $\{0,1,4,5,10,11,14,15\}.$ 
\end{tabular}
\end{center}

Clearly, as our xor-mask $s$ must satisfy all previous equations, we can deduce that 
$$
s \in \Omega_1 \cap \Omega_2 \cap \Omega_3 = \{0,5\}
$$ 
and, as $s \neq 0$ by definition, we can conclude that our xor-mask is, in fact, $s = 5$. Note that as the process of obtaining the different values for $\tilde \omega$ is non-deterministic, the remaining question is this: how many times do I have to execute Simon's algorithm in order to obtain such a system with enough probability? 

\begin{thm}
Simon's algorithm finds the correct solution for Simon's problem in $\mathcal{O}(n)$ steps with probability greater than $1/3$.
\end{thm}

\begin{proof}
Let us suppose that we have obtained $m$ linearly independent equations, namely with $\tilde \omega_1,\ldots,\tilde \omega_{m}$. Then, the probability of obtaining another linearly independent equation in the next interation of Simon's algorithm is $$\dfrac{2^n-2^m}{2^n}.$$ Thus, assuming $n \geq 3$, the probability of obtaining $n-1$ linearly independent equations after $n-1$ iterations of Simon's algorithm is $$
P = \left(1 - \dfrac{1}{2^n} \right)\left(1 - \dfrac{2}{2^n} \right) \cdots \left(1 - \dfrac{2^{n-2}}{2^n} \right) \geq \left( 1 - \sum_{k=2}^{n} \dfrac{1}{2^k} \right) \geq \dfrac{2^{n-1} - 1}{2^n} > \dfrac{1}{3}.
$$ 
\end{proof}

Even though Simon's algorithm is of little practical use in precisely the same way as Deutsch's and Deutsch-Jozsa's are, it shows once more that there exist problems such that a quantum computer is capable of solving them efficiently while a classical one is not. In fact, Simon's algorithm shows that there exist problems such that a quantum Turing machine is exponentially faster than a probabilistic Turing machine (PTM) \cite{nielsen2002quantum}. The difference between Simon's and Deutsch-Jozsa is that the latter can be solved by a PTM with an arbitrarily small error, while the former would take an exponential time to solve with such a machine.\\ 

Finally, although Simon's algorithm stablishes an oracle separation between BPP (Bounded-Error Probabilistic Polynomial time classical algorithms) and BQP, we still do not know if BPP $\neq$ BQP, as Simon's problem depends on a black box.\\

\newpage

\section{Shor's Factoring Algorithm}\label{sec::shor}

Let us begin with a problem in number theory that depends on one of the most well known theorems of all time:

\begin{defi}
Let $N \in \mathbb{Z}_{\geq 0}$, the fundamental theorem of arithmetic tells us that there exists a unique factorization of $N$ as a product of prime powers: $$ N = p_1^{\alpha_1}p_2^{\alpha_2} \cdots p_k^{\alpha_k} = \prod_{i=1}^{k}p_i^{\alpha_i}.$$ The prime factorization problem, or {\em PFP}, is the problem of finding such a factorization for a given number $N \in \mathbb{Z}_{\geq 0}$.
\end{defi}

Many mathematicians have worked on algorithms that calculate the prime factorization of an integer. To understand the ideas behind the most recent solutions to this problem, we must go back to the 17th Century, when Pierre de Fermat invented an elegant factorization method that today bears his name. Fermat's method consists in representing an odd number $N$ as a difference of squares, which is easily proven to exist. Then, as $N = n^2 - m^2 = (n + m)(n - m)$, we have that $\gcd(n + m, N)$ and $\gcd(n - m, N)$ are hopefully non-trivial factors of $N$. \\

But it was not until the beginning of the 20th Century that some improvements to this method were made, as mathematicians like Maurice Kraitchik in 1922 \cite{kraitchik1924recherches}, Derrick Henry Lehmer and Ralph Ernest Powers in 1931 \cite{lehmer1931factoring}, Michael A. Morrison and John Brillhart in 1975 \cite{morrison1975method} and Richard Schroeppel at the end of the 1970s (unpublished, but described in \cite{pomerance1982analysis}) developed factorization methods whose ideas were around the original Fermat's method. These upgrades eventually arrived at its maximum expression with the Quadratic Sieve developed by Carl Pomerance in 1984 \cite{pomerance1984quadratic} and the General Number Field Sieve, a collective effort based essentially on ideas due to John Pollard in 1993 \cite{lenstra1993number,bernstein1993general}. For more information about the story behind the evolution of Fermat's idea, see \cite{pomerance1996tale}. \\

The two previous methods are currently the most efficient classical algorithms for factoring an integer. However, they still have a problem: their computational complexity is super-polynomial in $\log N$, the number of digits in $N$. In fact, the GNFS, which has proven to be the most efficient known classical algorithm for factoring integers larger than $10^{100}$, has (heuristic) computational order
$$ 
\mathcal{O} \left(e^{ (\log N)^{\frac{1}{3}}(\log \log N)^{\frac{2}{3}}}\right).
$$

Unfortunately, these methods have the constraints of any super-polynomial algorithm and, at the present time, no known ponynomial-time classical algorithm exists for the factoring problem (i.e., PFP is not known to be in P). Nevertheless, verifying that a candidate solution for this problem is in fact the actual solution is computationally easy; thus, PFP is in NP. \\

There is a very well known result in computational complexity theory, due to Richard E. Ladner \cite{ladner1975structure}, that tells us the following: if P $\neq$ NP, then there exists a non-empty class, called NP-intermediate, that contains all problems in NP which are neither in the class P nor in NP-complete. It is widely believed that PFP is inside this class. \\

What we surely know, thanks to Peter W. Shor and its acclaimed polynomial-time quantum algorithm for prime factorization \cite{shor1994algorithms}, is that the PFP is in BQP. The objective of this section is to describe such result. For that, we shall first define Shor's algorithm as a classical one that relies on a black box that finds the multiplicative order of a certain integer modulo $N$. Next, we will provide a quantum algorithm that substitutes that black box. Finally, we will describe its performance via a worked-out example.

\subsection*{The classical part}

Let $N \in \mathbb{Z}_{\geq 0}$, we proceed to define Shor's algorithm for factoring $N$.

\paragraph{$\mathbb{STEP}$ 1}
$$$$\noindent\framebox{\parbox[b]{\linewidth}{\begin{algorithmic}
\State $x \leftarrow$ random integer such that $1 < x < N$
\State $d \leftarrow \text{gcd}(x, N)$
\end{algorithmic}}} \\

It is clear that, if $d > 1$, we have already found a factor of $N$. However, the probability of such an unlikely event is small, and in this case we proceed to next step. Note that the computational complexity of this step (i.e., of calculating the greatest common divisor of $x$ and $N$) has order $\mathcal{O}(\log^2 N)$ \cite{knuth1968art}.

\paragraph{$\mathbb{STEP}$ 2}
$$$$\noindent\framebox{\parbox[b]{\linewidth}{\begin{algorithmic}
\State $r \gets O_N(x)$
\end{algorithmic}}} \\

This is the step that we shall resolve with the aid of a quantum computer, as will be explained later. For now, let us recall that the multiplicative order of $x$ modulo $N$, provided that $\gcd(x,N)=1$, is defined as $$O_N(x) = \min \{r \in \mathbb{Z}_{> 0} : x^r \equiv 1 \mod N \}.$$ Calculating the multiplicative order is a hard problem in the general case, and the best known classical algorithm that solves it has a super-polynomial computational complexity \cite{cohen2013course}. \\

Right now, we have obtained a certain $r$ such that $r = O_N(x)$. However, not any value of $r$ serves our purposes. At the end of this step, we shall check if $r$ is an even number and, if that holds, we have to also check if $x^{r/2} + 1 \not\equiv 0 \mod N$. If any of those two conditions fail, we shall go back to the beginning of the algorithm, and repeat it again with a different random value for $x$. The unavoidable question is: what is the probability of not making it?

\begin{thm}

Let $N \in \mathbb{Z}_{\geq 0}$ such that $2 \nmid N$ and whose prime factorization is $$ N = p_{1}^{\alpha_1}p_{2}^{\alpha_2}\cdot\cdot\cdot p_{k}^{\alpha_k}.$$ Suppose $x$ is chosen at random, with $1 < x < N$ and $\gcd(x,N) = 1$, and let $r = O_N (x)$. Then:

{\em $$\text{Prob \Big[} (2 \mid r) \wedge (x^{r/2}+1 \not\equiv 0 \text{ mod } N) \text{\Big]} \geq 1 - \frac{1}{2^{k-1}}  $$ }

\end{thm}

\begin{proof}

\cite{ekert1996quantum} Appendix B.

\end{proof}

In other words, the probability of obtaining a number $x$ that fulfills all the conditions of the algorithm is greater than 1/2 in the worst case (i.e., when $N$ has only two different prime factors).

\paragraph{$\mathbb{STEP}$ 5}
$$$$\noindent\framebox{\parbox[b]{\linewidth}{\begin{algorithmic}
\State  $d_1 \gets \text{gcd}(x^{r/2} + 1, N)$
\State $d_2 \gets \text{gcd}(x^{r/2} - 1, N)$
\end{algorithmic}}} \\

As $2 \mid r$ and $x^{r/2} + 1 \not\equiv 0 \mod N$, it is easy to see that, from 
$$
x^r - 1 \equiv (x^{r/2} - 1)(x^{r/2} + 1) \equiv  0 \text{ mod } N,
$$ 
we can conclude that $d_1$ and $d_2$ are non-trivial factors of $N$, thus accomplishing the main purpose of the algorithm. As promised, what remains to be seen is the calculus of the multiplicative order of $x$ modulo $N$ with the help of a quantum computer. We describe now this process.

\subsection*{The quantum part}

We are given two non--negative integers $N$ and $x$, with $1<x<N$. Our aim is finding the order of $x$ in the congruence ring ${\mathbb Z} / N {\mathbb Z}$.

\paragraph{$\mathbb{SETUP}$}
$$$$\noindent\framebox{\parbox[b]{\linewidth}{\begin{algorithmic}
\State $ \Ket{\psi_0}_{t,n} \leftarrow \Ket{0}_t \otimes \Ket{0}_n$
\end{algorithmic}}} \\

First, we need a quantum computer with two registers of sizes $t$ and $n$ respectively, where $n = \ceil*{\text{log}_2 N} $ and $t = 2n$ (the reason behind this will be clear later). All qubits are initialized at 0.

\paragraph{$\mathbb{STEP}$ 2.1}
$$$$\noindent\framebox{\parbox[b]{\linewidth}{\begin{algorithmic}
\State $\Ket{\psi_1}_{t,n} \leftarrow (\boldsymbol{H}^{\otimes t} \otimes \boldsymbol{I}^{\otimes n}) \left( \Ket{\psi_0}_{t,n} \right)$
\end{algorithmic}}} \\

This transformation is now a common factor of our quantum algorithms and there is no need of explaining it furthermore. The crucial point is that after its application the first register is in a superposition of all states of the computational basis with equal amplitudes given by $1 / \sqrt{2^t}$. More precisely:
$$ 
\Ket{\psi_1}_{t,n} = \frac{1}{\sqrt{2^t}} \sum_{j=0}^{2^t - 1} \Ket{j}_t \otimes \Ket{0}_n
$$

\paragraph{$\mathbb{STEP}$ 2.2}
$$$$\noindent\framebox{\parbox[b]{\linewidth}{\begin{algorithmic}
\State $\Ket{\psi_2}_{t,n} \leftarrow \boldsymbol{M}_{x,N} \left( \Ket{\psi_1}_{t,n} \right)$
\end{algorithmic}}} \\

Let $n,t,x$ and $N$ defined as in the context of this algorithm, the modular exponentiation gate is the unitary operator that has the following effect on the basis states of a quantum system: $$\boldsymbol{M}_{x,N} : \Ket{j}_t \otimes \Ket{k}_n \rightarrow \Ket{j} \otimes \ket{k + x^j \text{ mod } N}.$$ This transformation is unitary, and its construction takes $\mathcal{O}(\log^3 N)$ steps \cite{shor1994algorithms}. Thus, as will be clear at the end of this subsection, it represents the bottleneck of Shor's algorithm. 
$$
\ket{\psi_2}_{t,n} = \boldsymbol{M}_{x,N} \left( \Ket{\psi_1}_{t,n} \right) = \frac{1}{\sqrt{2^t}} \sum_{j=0}^{2^t - 1} \boldsymbol{M}_{x,N} \left( \Ket{j}_t \otimes \Ket{0}_n \right) = \frac{1}{\sqrt{2^t}} \sum_{j=0}^{2^t - 1} \Ket{j}_t \otimes \Ket{x^j \text{ mod } N}_n.
$$

Thanks again to quantum parallelism, we have now generated all powers of $x$ modulo $N$ simultaneously. From now on, in order to make the tracking of the algorithm clearer, we shall suppose that $r$ is a power of 2. In this case, our current quantum state can be expressed as follows:

$$ 
\Ket{\psi_2}_{t,n} = \dfrac{1}{\sqrt{2^t}} \sum_{b=0}^{r-1} \left[ \left( \sum_{a=0}^{\frac{2^t}{r} - 1} \ket{ar + b}_t \right) \otimes \ket{x^b \text{ mod } N}_n \right].
$$

The general case where $r$ may not be a power of 2 is more difficult to express, and is better to explain it in the part devoted to the example.

\paragraph{$\mathbb{STEP}$ 2.3}
$$$$\noindent\framebox{\parbox[b]{\linewidth}{\begin{algorithmic}
\State $\tilde \delta \leftarrow$ measure the second register
\State
\State $ \Ket{\psi_3}_{t} \leftarrow \Ket{\psi_2}_{t,n}$ after measuring the second register
\end{algorithmic}}} \\

Let us suppose that, for a certain $b_0 \in \{0, \ldots, r-1 \}$, we obtain the value $\tilde \delta = x^{b_0}$ mod $N$. Thus, the computer is now in the following quantum state: $$ \Ket{\psi_3}_t = \sqrt{\frac{r}{2^t}}  \sum^{\frac{2^t}{r} - 1}_{a=0} \Ket{ar + b_0}_t.$$ Note that now we only have $2^t/r$ terms in the sum, instead of the previous $2^t$ ones, and that the value we are looking for (i.e., $r$) is beginning to surface inside our quantum system in the form of a period. The next step will provide us with a tool capable of extracting this period from a quantum state.

\paragraph{$\mathbb{STEP}$ 2.4}
$$$$\noindent\framebox{\parbox[b]{\linewidth}{\begin{algorithmic}
\State $ \ket{\psi_4}_t \leftarrow \boldsymbol{F}_n \left( \ket{\psi_3}_t \right)$
\end{algorithmic}}} \\

This step requires that we define a new quantum gate: the quantum discrete Fourier transform, or QFT, whose effect on the quantum basis states is: 
$$
\boldsymbol{F}_n : \ket{j}_n \longmapsto \dfrac{1}{\sqrt{2^n}} \sum_{k=0}^{2^n-1} e^{2\pi i \cdot jk/2^n} \ket{k}_n
$$

After applying the QFT to the first register, we can express the obtained state as follows:
\begin{eqnarray*}
\ket{\psi_4}_t & = & \sqrt{\frac{r}{2^t}}  \sum^{(2^t/r) - 1}_{a=0} \boldsymbol{F}_t (\Ket{ar + b_0}_t)
 = \sqrt{\frac{r}{2^t}} \sum^{(2^t/r)-1}_{a=0} \left( \frac{1}{\sqrt{2^t}} \sum^{2^t-1}_{j=0} e^{2\pi i \cdot j(ar+b_0)/2^t} \ket{j}_t \right ) \\
 & = & \frac{1}{\sqrt{r}} \left( \sum^{2^t-1}_{j=0} \left[ \frac{r}{2^t} \sum^{(2^t/r)-1}_{a=0} e^{2\pi i \cdot ja/(2^t/r)} \right] e^{2\pi i \cdot jb_0/2^t} \ket{j}_t \right)
\end{eqnarray*}

Now, in order to arrive at the last version of the state, we use the following property of the exponential sums: 
$$ 
\dfrac{1}{2^t/r} \sum\limits^{\frac{2^t}{r}-1}_{a=0} e^{2 \pi i \cdot ja/(2^t/r)} = 
\left\{ \begin{array}{l} 1 \text{ if } (2^t/r) \mid j \\ 0 \text{ otherwise,} \end{array} \right.
$$ 
then the only non-zero terms in the bracket sum are those with the form $j=2^tk/r$. Then we can rewrite the state as:
$$ 
\ket{\psi_4}_t = \dfrac{1}{\sqrt{r}}\left( \sum\limits^{r-1}_{k=0} e^{2 \pi i \cdot kb_0/r} \Ket{\dfrac{k2^t}{r}}_t \right) 
$$

It is time now to measure the first register.

\paragraph{$\mathbb{STEP}$ 2.5}
$$$$\noindent\framebox{\parbox[b]{\linewidth}{\begin{algorithmic}
\State $\tilde \omega \leftarrow$ measure $\Ket{\psi_4}_t$
\end{algorithmic}}} \\

We obtain, for some unknown $k_0 \in \{0,\ldots,r-1\}$ (the probability of obtaining each of them is the same, but this changes in the general case where $r$ may not be a power of 2): 
$$
\tilde\omega = \frac{k_0 2^t}{r}
$$

If $\tilde\omega = 0$, we obtain no information and we must come back to the beginning. If, otherwise, $\tilde\omega \neq 0$, we can obtain some information about $r$ (or the actual value of $r$ indeed), just by putting $\tilde \omega / 2^t$ as a fraction in lowest terms and taking the denominator. If we call $c$ that denominator and $x^c \not\equiv 1 \mod N$, then $c$ is a factor of $r$. If, on the other hand $x^c \equiv 1 \mod N$, then $r = c$ and we have obtained the multiplicative order of $x$ modulo $N$. In the former case, we should rerun the quantum part of Shor's algorithm with $x^c$ instead of with $x$, and by repeating this process a finite number of times we shall end up with the correct value of $r$. \\

Now let us show with an example how Shor's algorithm works for a particular input. 

\subsection*{Example: Factoring 217}

We want to find the prime factors of 217 using Shor's factoring algorithm. First, the algorithm chooses a random integer $x$ such that $1 < x < 217$. Let us suppose that we obtain $x = 5$, which fulfills the first condition: $\text{gcd}(5, 217) = 1$. Next we find the multiplicative order of 5 modulo 217, which is achieved with the help of a quantum computer. \\

We calculate $t$ and $n$ and initialize the quantum system. In this case, $n = \ceil*{\log_2 217} = 8$ and consequently $t = 2n = 16$. Thus, our quantum registers have the following initial states: \\

\noindent\framebox{\parbox[b]{\linewidth}{\begin{algorithmic}
\State $ \ket{\psi_0}_{16,8} \leftarrow \ket{0}_{16} \otimes \ket{0}_{8}$
\end{algorithmic}}} \\

Next, we apply the Hadamard transformation, hence obtaining a superposition of all basis states in the first register, all with identical amplitudes. This way, all integers between $0$ and $2^{16}-1$ are now \textit{somewhere} in the first register. \\

\noindent\framebox{\parbox[b]{\linewidth}{\begin{algorithmic}
\State $ \Ket{\psi_1}_{16,8} \leftarrow ( \boldsymbol{H}^{\otimes 16} \otimes \boldsymbol{I}^{\otimes 8} ) \left( \ket{\psi_0}_{16,8} \right)$
\end{algorithmic}}}

$$ 
\Ket{\psi_1}_{16,8} = \frac{1}{\sqrt{2^{16}}} \sum_{j=0}^{2^{16} - 1} \Ket{j}_{16} \otimes \Ket{0}_8
$$

Afterwards, we apply the quantum gate that calculates the powers of 5 modulo 217, which has the following effect. \\

\noindent\framebox{\parbox[b]{\linewidth}{\begin{algorithmic}
\State $\Ket{\psi_2}_{16,8} \leftarrow \boldsymbol{M}_{5,217} \left( \Ket{\psi_1}_{16,8} \right)$
\end{algorithmic}}} \\ 

$$
\ket{\psi_2}_{16,8} = \frac{1}{\sqrt{2^{16}}} \sum_{j=0}^{2^{16} - 1} \boldsymbol{M}_{5,217} \left( \Ket{j}_{16} \otimes \Ket{0}_{8} \right) = \frac{1}{\sqrt{2^{16}}} \sum_{j=0}^{2^{16} - 1} \Ket{j}_{16} \otimes \Ket{5^j \text{ mod } 217}_8
$$

If we expand the sum, we can alternatively express the result as follows (please note that, for the sake of simplicity, we have omitted the subindices and the tensor product operators):
$$ 
\begin{array}{ l @{\hspace{2bp}} r @{\hspace{2bp}} r @{\hspace{2bp}} r @{\hspace{2bp}} r @{\hspace{2bp}} r @{\hspace{2bp}} r @{\hspace{2bp}} r @{\hspace{2bp}} r @{\hspace{2bp}} r @{\hspace{2bp}} r @{\hspace{2bp}} r @{\hspace{2bp}} l}
\Ket{\psi_2} = \frac{1}{\sqrt{2^{16}}} & ( \ket{0}\ket{1} & + & \ket{1}\ket{5} & + & \ket{2}\ket{25} & + &  \ket{3}\ket{125} & + & \ket{4}\ket{191}   & + & \ket{5}\ket{87} & +\\
                                                           & \ket{6}\ket{1}   & + & \ket{7}\ket{5} & + & \ket{8}\ket{25} & + & \ket{9}\ket{125}  & + & \ket{10}\ket{191} & + & \ket{11}\ket{87} & + \\
                                                           & \ket{12}\ket{1}   & + & \ket{13}\ket{5} & + & \ket{14}\ket{25} & + & \ket{15}\ket{125}  & + & \ket{16}\ket{191} & + & \ket{17}\ket{87} & + \\
                                                           & \ket{18}\ket{1}   & + & \ket{19}\ket{5} & + & \ket{20}\ket{25} & + & \ket{21}\ket{125}  & + & \ket{22}\ket{191} & + & \ket{23}\ket{87} & + \\
                                                           & \ket{24}\ket{1}   & + & \ket{25}\ket{5} & + & \ket{26}\ket{25} & + & \ket{27}\ket{125}  & + & \ket{28}\ket{191} & + & \ket{29}\ket{87} & + \\           
                                                           &\ket{30}\ket{1}& +&\ldots)\,\,\,\,\,\,\,\,\,\,\,&&&&&&&&& 
\end{array}
$$

After a close inspection, we can observe that the values on the second register are periodic. If we take common factor, we end up with the state:
$$ 
\begin{array}{ l @{\hspace{2bp}} r @{\hspace{2bp}} r @{\hspace{2bp}} r @{\hspace{2bp}} r @{\hspace{2bp}} r @{\hspace{2bp}} r @{\hspace{2bp}} r @{\hspace{2bp}} r @{\hspace{2bp}} l @{\hspace{2bp}} l}
\Ket{\psi_2} = \frac{1}{\sqrt{2^{16}}} & \big(( \ket{0}   & + & \ket{6}    & + & \ket{12} & + &  \ket{18} & + \ldots +  & \ket{65526} + \ket{65532})\ket{1} + \\
                                                         & ( \ket{1}   & + & \ket{7}    & + & \ket{13} & + &  \ket{19} & + \ldots +  & \ket{65527} + \ket{65533})\ket{5} + \\
                                                         & ( \ket{2}   & + & \ket{8}    & + & \ket{14} & + &  \ket{20} & + \ldots +  & \ket{65528} + \ket{65534})\ket{25} + \\
                                                         & ( \ket{3}   & + & \ket{9}    & + & \ket{15} & + &  \ket{21} & + \ldots +  & \ket{65529} + \ket{65535})\ket{125} + \\
                                                         & ( \ket{4}   & + & \ket{10}  & + & \ket{16} & + &  \ket{22} & + \ldots +  & \ket{65530})\ket{191} + \\
                                                         & ( \ket{5}   & + & \ket{11}  & + & \ket{17} & + &  \ket{23} & + \ldots +  & \ket{65531})\ket{87}\big)\\
\end{array}
$$

Thanks to this representation, it is easier to understand what will happen if we measure the second register. \\

\noindent\framebox{\parbox[b]{\linewidth}{\begin{algorithmic}
\State $\tilde \delta \leftarrow$ measure the second register
\State
\State $ \Ket{\psi_3}_{16} \leftarrow \Ket{\psi_2}_{16,8}$  after measuring the second register
\end{algorithmic}}} \\

It is clear that we shall obtain a value $\tilde \delta$ such that $\tilde \delta = 5^{\, \tilde j} \text{ mod } 217$ for a certain $\tilde j \in \{0,\ldots,2^{16}-1\}$. Thus, $\tilde \delta \in \{1,5,25,125,191,87\}$ (which are the powers of 5 modulo 217). Let us suppose that we get $\tilde \delta = 25$. The register will collapse into $\ket{25}$ and all other possible values will be destroyed and gone forever, the information about the rest of possible powers of 5 modulo 217 lost and the first register will also collapse into the values that were tensored with $\ket{25}$, thus discarding the remaining ones. What interests us is that all the basis states tensored with $\ket{25}$ correspond with the exponents $\tilde j$ such that $\tilde \delta = 5^{\tilde j} \text{ mod } 217$. More specifically:
$$
\Ket{\psi_3}_{16} = \frac{1}{\sqrt{10923}} \Big( \ket{2}_{16} + \ket{8}_{16} + \ket{14}_{16} + \ldots + \ket{65528}_{16} + \ket{65534}_{16} \Big) = \frac{1}{\sqrt{10923}}\left( \sum^{10922}_{a=0} \ket{6a + 2}_{16} \right)
$$

A pattern has arisen, as the basis states on the first register display some periodic behavior. This period naturally corresponds with the multiplicative order of 5 modulo 217, which happens to be equal to 6. Of course, this information is yet hidden to us, we are just using some knowledge of the problem for providing a mathematical explanation of the performance of the algorithm in this particular case. On the other hand, the constant 10923 is just a normalization of the amplitudes after the collapse of the register, corresponding to the total number of exponents that return 25 modulo 217 between 0 and $2^{16}-1$. \\

As explained previously, if we want to obtain the period from within the bowels of our quantum system, we should use the quantum Fourier transform: \\

\noindent\framebox{\parbox[b]{\linewidth}{\begin{algorithmic}
\State $ \ket{\psi_4}_{16} \leftarrow \boldsymbol{F}_{16} \left( \ket{\psi_3}_{16} \right)$
\end{algorithmic}}} \\

\begin{eqnarray*}
\ket{\psi_4}_{16} & = & \boldsymbol{F}_{16} \left(\frac{1}{\sqrt{10923}} \sum^{10922}_{a=0} \ket{6a + 2}_{16} \right) = \frac{1}{\sqrt{10923}} \sum^{10922}_{a=0}  \boldsymbol{F}_{16} \big( \ket{6a + 2}_{16} \big) \\
& = & \frac{1}{\sqrt{10923}} \sum^{10922}_{a=0} \left( \dfrac{1}{\sqrt{2^{16}}} \sum_{k=0}^{2^{16}-1} e^{2 \pi i \cdot (6a+2) k/2^{16}} \ket{k}_{16}\right) \\
& = & \frac{1}{\sqrt{2^{16}}} \sum^{2^{16}-1}_{k=0} \left( \left[ \frac{1}{\sqrt{10923}} \sum^{10922}_{a=0} e^{2\pi i \cdot 6ak/2^{16}} \right] e^{2\pi i \cdot 2k/2^{16}} \ket{k}_{16} \right)
\end{eqnarray*}

Thus, we end up again with a distribution of the amplitudes defined by 
$$
\alpha_k = \dfrac{1}{\sqrt{2^{16} \cdot 10923}} \left( \sum^{10922}_{a=0} e^{2\pi i \cdot 6ak/2^{16}} \right)  e^{2\pi i \cdot 2k/2^{16}},
$$
which gives us the following probability distribution (not uniform, as the order in this case is not a power of 2):
$$ 
|\alpha_k|^2 = \dfrac{1}{2^{16}\cdot 10923} \left\lvert \sum^{10922}_{a=0} e^{2\pi i \cdot 6ja/2^{16}} \right\rvert^2
$$ 
with $k=0,\ldots,2^{16}-1$. If we represent this distribution in a graph, it is easier to understand the final steps of the algorithm (please note that each of peaks is composed of many numbers with probability greater than 0, not just of a single integer).

\begin{center}
\includegraphics[width=0.9\textwidth]{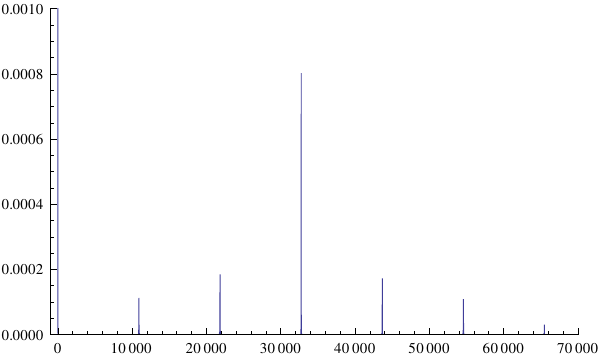}
\end{center}
   
$$$$\noindent\framebox{\parbox[b]{\linewidth}{\begin{algorithmic}
\State $ \tilde \omega \leftarrow  \text{ measure } \Ket{\psi_4}_{16}$
\end{algorithmic}}} \\

We then measure the first register, and obtain non-deterministically a value from one of the seven peaks. The first peak yields a 0, thus forcing us to start again the quantum part of the algorithm. If, otherwise, we obtain a value from the other six peaks, we can retrieve (for example, using continued fractions) the order $r$ from it. \\ 

Let us explain this last step more formally (we took details for this from \cite{lomonaco2002shor}). We recall that a continued fraction is represented as 
$$
[a_0;a_1,\ldots,a_K] = a_0+\cfrac{1}{a_1+\cfrac{1}{a_2+\cfrac{1}{a_3+\cfrac{1}{\ddots+\cfrac{1}{a_K}}}}}
$$ 
where $a_0 \in \mathbb{Z}_{\geq 0}$ and $a_1 , \ldots , a_K \in \mathbb{Z}_{> 0}$. Let $[a_0;a_1,\ldots,a_K]$, then there exists a unique $q \in \mathbb{Q}_{> 0}$ such that $q = [a_0;a_1,\ldots,a_K]$. We define the $k$-th convergent of $[a_0;a_1,\ldots,a_K]$, where $0 \leq k < K$, as $$q_k = [a_0;a_1\ldots,a_k].$$ Provided a certain $q \in \mathbb{Q}_{> 0}$, we can compute the members of its corresponding continued fraction as follows:
\begin{equation*}
\begin{split}
a_0 \leftarrow \lfloor q \rfloor, & \quad q_0 \leftarrow q - a_0 \\
\vdots \\
a_{k+1} \leftarrow \lfloor 1 / q_k \rfloor, & \quad q_{k+1}  \leftarrow 1 / q_k - a_{k+1}
\end{split}
\end{equation*} 
Even more, each of the convergents $q_k$ can be expressed as $q_k = b_k / c_k$, where $\gcd (b_k, c_k) = 1$ and 
\begin{equation*}
\begin{split}
b_0 \leftarrow a_0, & \quad c_0 \leftarrow 1 \\
b_1 \leftarrow a_0 a_1 + 1, & \quad c_1 \leftarrow a_1 \\
\vdots \\
b_{k+2} \leftarrow a_{k+2} b_{k+1} + b_k, & \quad c_{k+2} \leftarrow a_{k+2} c_{k+1} + c_k
\end{split}
\end{equation*}

It can be proven \cite{lomonaco2002shor} that we can obtain the period $r$ from $\tilde{w}$ with the following algorithm: starting at $k=1$, we compute $b_k$ and $c_k$ as previously explained from the continued fraction of $q = \tilde{w} / 2^n$. Then, we check if $$x^{c_k} \equiv 1 \mod N.$$ If the answer is affirmative, we have obtained the order $r$; if not, we do it again for $k+1$. \\

For example, a possible value from the second peak is 10915, and if we divide it by $2^{16}$ we obtain the continued fraction $$\dfrac{10915}{65536} = [0; 6, 237, 3, 1, 1, 6],$$ whose first convergent is $\frac{1}{6}$. Thus, as $5^6 \equiv 1 \mod 217$ we have that the order is 6. There is a probability of not obtaining $r$ in this step (not all peaks will lead to a successful result), thus having to restart the algorithm again to obtain the correct order. As shown in \cite{lomonaco2002shor}, the probability of success has a lower bound of $$\dfrac{0.232}{\log \log N} \left(1 - \dfrac{1}{N} \right)^2.$$ \\

It remains to be seen if $x=5$ and $r=6$ fulfill the final conditions of the algorithm. As 6 is an even number, and $5^{6/2} + 1 = 126 \neq 0 \text{ mod } 217$, we can finally proceed to the last step. As explained, we can calculate now two non-trivial factors of 217: 
$$
d_1 = \gcd(126, 217) = 7, \qquad d_2 = \gcd(124, 217) = 31,
$$ 
which in fact are the only prime factors of 217. Thus ends Shor's algorithm for this particular case. \\

\noindent\framebox{\parbox[b]{\linewidth}{\begin{algorithmic}
\State $ 217 = 7\times 31$
\end{algorithmic}}} \\

The first experimental demonstration of Shor's factoring algorithm came in 2001, \cite{vandersypen2001experimental} when a group at IBM factored 15 into 3 and 5 using a nuclear magnetic resonance (NMR) quantum computer with seven spin-1/2 nuclei in a molecule as qubits. In 2007, Shor's algorithm was implemented with photonic qubits by two different groups \cite{lanyon2007experimental,lu2007demonstration}, with both of them observing quantum entanglement in the process. In 2012, number 143 was factored with the help of an adiabatic quantum computer \cite{xu2012quantum}.

\newpage

\section{Grover's Search Algorithm}\label{sec:grover}

Grover's search algorithm was first described by Lov K. Grover in \cite{grover1996fast, grover1997quantum, grover2001schrodinger}, hence its name. The original aim of Grover's algorithm is the following: we have an unstructured and disorganized database with $2^n$ elements, identified from now on with the indices $0, \ldots , 2^n -1 $, and we want to find one that satisfies a certain property. The algorithm leans on two hypotheses: first, that such an element exists inside the database; and second, that this element is unique. \\

If we are to search for this target element with a classical computer, in the worst case we will have to check all members of the database, which tells us that this problem has a classical computational complexity of $\mathcal{O}(2^n)$. As will be shown, Grover's quantum search algorithm will make only $\mathcal{O}(\sqrt{2^n})$ queries to the database, thus strictly improving the performance of its classical counterpart (not with an exponential improvement though).

\paragraph{$\mathbb{SETUP}$}
$$$$\noindent\framebox{\parbox[b]{\linewidth}{\begin{algorithmic}
\State $\ket{\psi_0}_{n,1} \leftarrow \ket{0}_n \otimes \ket{1}$
\end{algorithmic}}} \\

We need a quantum computer with $n+1$ qubits, where the first $n$ qubits will be initialized at $\ket{0}$ and the last one at $\ket{1}$.

\paragraph{$\mathbb{STEP}$ 1}
$$$$\noindent\framebox{\parbox[b]{\linewidth}{\begin{algorithmic}
\State $\Ket{\psi_1}_{n,1} \leftarrow \boldsymbol{H}^{\otimes n+1} \left( \Ket{\psi_0}_{n,1} \right)$
\end{algorithmic}}} \\

In the first proper step of the algorithm, we apply the Hadamard quantum gate to all the qubits in our system. This way, the obtained result is a combination of all basis states inside our first $n$ qubits, and the $\ket{-}$ state in the remaining one.
$$
\ket{\psi_1}_{n,1}  = \Big(\boldsymbol{H}^{\otimes n} \ket{0}_n \Big)  \otimes  \Big( \boldsymbol{H} \ket{1} \Big) = \Big( \boldsymbol{H} \ket{0} \Big)^{\otimes n}  \otimes \Big( \boldsymbol{H} \ket{1} \Big) = \Big( \dfrac{1}{\sqrt{2^n}} \sum_{j=0}^{2^n-1} \ket{j}_n \Big) \otimes \ket{-}
$$

For the sake of simplicity, we reintroduce a useful quantum state that will appear many times throughout the rest of the algorithm. It will also be needed in order to define one of the key quantum gates of Grover's search method, as will be shown in the next step.
$$
\ket{\gamma}_n = \dfrac{1}{\sqrt{2^n}} \sum_{j=0}^{2^n -1} \ket{j}_n 
$$

This definition helps us to express our quantum state as $\ket{\psi_1}_{n,1} = \ket{\gamma}_n \otimes \ket{-}$.

\paragraph{$\mathbb{STEP}$ 2.1}
$$$$\noindent\framebox{\parbox[b]{\linewidth}{\begin{algorithmic}
\State $\Ket{\psi_2}_{n,1} \leftarrow \boldsymbol{O}_{f} \left( \Ket{\psi_1}_{n,1} \right)$
\end{algorithmic}}}

\paragraph{}The next step of the algorithm leans on the following assumption: we can build a quantum gate, called $\boldsymbol{O}_{f}$, that makes use of a function capable of recognizing the element of the database we are searching for. This function $f$ can be defined as follows (note that the desired element is identified by the unknown index $j_0 \in \{ 0 , \ldots , 2^n-1 \}$):
\[
  f(j) =
  \begin{cases}
        1 & \text{if $j=j_0$} \\
        0 & \text{otherwise}
  \end{cases}
\] \\

The $\boldsymbol{O}_{f}$ quantum gate is then the oracle gate described in previous algorithms, which has the following effect on the basis states of a $n+1$-qubit system 
$$
\boldsymbol{O}_{f} : \ket{j}_n \otimes \ket{k} \longmapsto \ket{j}_n \otimes \ket{k \oplus f(j)}.
$$ 

We recall from Deutsch-Jozsa algorithm that the oracle gate has the following effect on the state $\ket{j}_n \otimes \ket{-}$:
\begin{equation*}
\begin{split}
\boldsymbol{O}_{f} \Big( \ket{j}_n \otimes \ket{-} \Big) & =  (-1)^{f(j)} \ket{j}_n \otimes \ket{-}
\end{split}
\end{equation*}

This means that, in our case, $\boldsymbol{O}_{f}$ inverts the sign of the amplitude corresponding to the basis state that codifies the searched element in the first $n$ qubits, while keeping intact the rest of the amplitudes. \\

In order to make the explanation and understanding of the algorithm easier and simpler, we introduce another $n$-qubit quantum state, 
$$
\ket{\rho}_n = \dfrac{1}{\sqrt{2^n -1}} \sum_{\substack{j=0 \\ j\neq j_0}}^{2^n-1} \ket{j}_n,
$$ 
whose relationship with the previously described state $\ket{\gamma}_n$ is $$\ket{\gamma}_n = \dfrac{\sqrt{2^n - 1}}{\sqrt{2^n}} \ket{\rho}_n + \dfrac{1}{\sqrt{2^n}} \ket{j_0}_n.$$ Please note that $\ket{\rho}$ depends on the value of $j_0$, but we shall write $\ket{\rho}$ instead of $\ket{\rho\,(j_0)}$ for the sake of simplicity. The introduction of the notation $\ket{\rho}_n$ helps in noticing the separation of the searched element $\ket{j_0}$ from the rest of the quantum basis states. \\

In fact, after applying $\boldsymbol{O}_{f}$ to our quantum system it will look like this:
\begin{eqnarray*}
\ket{\psi_2}_{n,1} & = & \boldsymbol{O}_{f} \Big( \ket{\gamma}_n \otimes  \ket{-} \Big) = \boldsymbol{O}_{f} \left( \left(\dfrac{\sqrt{2^n -1}}{\sqrt{2^n}} \ket{\rho}_n + \dfrac{1}{\sqrt{2^n}} \ket{j_0}_n \right) \otimes \ket{-} \right) \\
& = & \left( \dfrac{\sqrt{2^n -1}}{\sqrt{2^n}} \ket{\rho}_n - \dfrac{1}{\sqrt{2^n}} \ket{j_0}_n \right) \otimes \ket{-} = \left( \ket{\gamma}_n - \dfrac{2}{\sqrt{2^n}} \ket{j_0}_n \right) \otimes \ket{-}
\end{eqnarray*} \\

The motivation behind the last interpretation of $\ket{\psi_2}_{n,1}$ will become clear in the next step.

\paragraph{$\mathbb{STEP}$ 2.2}
$$$$\noindent\framebox{\parbox[b]{\linewidth}{\begin{algorithmic}
\State $\ket{\psi_3}_{n,1} \leftarrow (\boldsymbol{\varGamma}_n \otimes \boldsymbol{I}) \left(\ket{\psi_2}_{n,1}\right) $
\end{algorithmic}}} \\

For this step, we must construct a new quantum gate, denoted by $\boldsymbol{\varGamma}_n$, that will affect only the first $n$ qubits. The remaining qubit will remain intact (this is represented with the single-qubit identity gate $\boldsymbol{I}$). The definition of $\boldsymbol{\varGamma}_n$ is:

$$\boldsymbol{\varGamma}_n = 2\ket{\gamma}_n \bra{\gamma}_n - \boldsymbol{I}^{\otimes n}$$

Let us see what happens when we apply this new quantum gate to the first $n$ qubits of our quantum state $\ket{\psi_2}_{n,1}$ defined in the previous step.
\begin{eqnarray*}
\boldsymbol{\varGamma}_n \left( \ket{\psi_2}_{n} \right) & = & \Big( 2\ket{\gamma}_n \bra{\gamma}_n - \boldsymbol{I}^{\otimes n} \Big) \left( \ket{\gamma}_n - \dfrac{2}{\sqrt{2^n}} \ket{j_0}_n \right) \\
& = & 2 \ket{\gamma}_n \braket{\gamma | \gamma}_n - \dfrac{4}{\sqrt{2^n}} \ket{\gamma}_n \braket{\gamma | j_0}_n - \ket{\gamma}_n + \dfrac{2}{\sqrt{2^n}} \ket{j_0}_n \\
& = & 2 \ket{\gamma}_n - \dfrac{4}{2^n} \ket{\gamma}_n - \ket{\gamma}_n + \dfrac{2}{\sqrt{2^n}} \ket{j_0}_n \\
& = & \dfrac{2^{n-2} - 1}{2^{n-2}} \ket{\gamma}_n + \dfrac{2}{\sqrt{2^n}} \ket{j_0}_n 
\end{eqnarray*}

Mind that, from the properties of the inner product $\braket{\gamma | \gamma}_n = 1$ and $\braket{\gamma | i_0}_n = 1/\sqrt{2^n}$. So, all the basic states but $\ket{j_0}$ have decreased uniformly the probability associated to their amplitude, while $\braket{j_0}$ has actually increased it accordingly.\\

The $\boldsymbol{\varGamma}_n$ gate is also called Grover diffusion operator, and can also be seen as 
$$
\boldsymbol{\varGamma}_n = \boldsymbol{H}^{\otimes n} \Big( 2\ket{0}_n \bra{0}_n - \boldsymbol{I}^{\otimes n} \Big) \boldsymbol{H}^{\otimes n}, 
$$ 
which is a much more painless way of implementing it in practice. Steps 2.1 and 2.2 are usually treated as a single step, represented by the Grover gate 
$$
\boldsymbol{G}_n = (\boldsymbol{\varGamma}_n \otimes \boldsymbol{I}) (\boldsymbol{O}_f).
$$

\paragraph{$\mathbb{STEP}$ 3(ish)}

$$$$\noindent\framebox{\parbox[b]{\linewidth}{\begin{algorithmic}
\State $\omega \leftarrow$  measure the first register
\end{algorithmic}}} \\

Now, in order to achieve the objective of Grover's algorithm, one must apply the $\boldsymbol{G}_n$ gate (that is, perform $\mathbb{STEP}$ 2) repeatedly until the probability of obtaining the index $j_0$ is maximal (a general overview of the algorithm can be seen in Figure \ref{grover-circuit}). 

After the desired probability has been reached, we measure the first $n$-qubit register, and obtain the index $j_0$. The optimal number of times $\boldsymbol{G}_n$ is applied has order $\mathcal{O}(\sqrt{2^n})$, this will be proved later. But first, let us explain the behavior of the algorithm with an example.

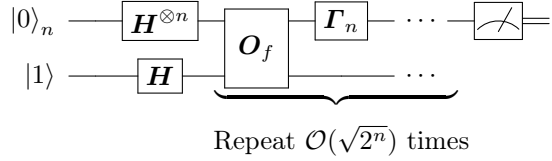
\begin{figure}
	\centering
\[
\Qcircuit @C=1em @R=.7em {
  \lstick{\ket{0}_n} & \qw & \gate{\boldsymbol{H}^{\otimes n}} & \multigate{1}{\boldsymbol{O}_f} & \gate{\boldsymbol{\varGamma}_n} & \qw & \cdots & & \meter & \cw \\
  \lstick{\ket{1}} & \qw & \gate{\boldsymbol{H}} & \ghost{\boldsymbol{O}_f} & \qw & \qw & \cdots & \\
  & & & & \dstick{\text{$\mathcal{O}(\sqrt{2^n})$ times}} \gategroup{1}{4}{2}{5}{.7em}{_\}} }
\]
	\caption{Circuit representation of Grover's Search Algorithm}
    \label{grover-circuit}
\end{figure}


\subsection*{An example with $n = 4$ qubits}

Let us illustrate how Grover's search works with this exemplifying case: suppose we have an unstructured database with its elements listed as $0,1, \ldots, 15$ indexed via 4 qubits, and that we want to find a certain item that is indexed with the number $7$ (note that we do not know this yet, but it can be supposed in order to explain the algorithm straightforwardly). 

As previously explained, we need a quantum computer with $n+1$ qubits, where $n$ is the number of bits needed for codifying the indices of the members of the database (in this case, $n = 4$). In order to correctly setup our quantum system, the first 4 qubits must be in the state $\ket{0}$ and the remaining one must be in the state $\ket{1}$. Thus, our quantum system begins as follows: 
$$
\ket{\psi_0}_{4,1} = \ket{0}_4 \otimes \ket{1}
$$ 

The first transformation we apply to the quantum system is the Hadamard gate $\boldsymbol{H}$, which converts it to the following state: 
$$
\ket{\psi_1}_{4,1} = (\boldsymbol{H}^{\otimes 4} \ket{0}_4) \otimes (\boldsymbol{H} \ket{1}) = \left( \dfrac{1}{4} \sum_{j=0}^{15} \ket{j}_4 \right) \otimes \ket{-}
$$

We remind the reader that from now on we will make use of the states $\ket{\gamma}_4$ and $\ket{\rho}_4$ in the interest of simplifying the writing of the whole process:
$$
\ket{\gamma}_4 = \dfrac{1}{4} \sum_{j=0}^{15} \ket{j}_4, \quad\quad 
\ket{\rho}_4 = \dfrac{1}{\sqrt{15}} \sum_{\substack{j=0 \\ j\neq 7}}^{15} \ket{j}_4, \quad\quad
\ket{\gamma}_4 = \dfrac{\sqrt{15}}{4} \ket{\rho}_4 + \dfrac{1}{4} \ket{7}_4
$$

The introduction of these states allows us to write $\ket{\psi_1}_{4,1}$ as:
$$
\ket{\psi_1}_{4,1} = \left( \dfrac{\sqrt{15}}{4} \ket{\rho}_4 + \dfrac{1}{4} \ket{7}_4 \right) \otimes \ket{-}
$$

On $\mathbb{STEP}$ 2.1 we make use of the transformation $\boldsymbol{O}_{f}$, which applies to all possible states inside the first 4 qubits an oracle $f$ that identifies 7 as the correct index for the element we are looking for.
$$
\ket{\psi_2}_{4,1} = \boldsymbol{O}_{f} \left[ \left( \dfrac{\sqrt{15}}{4} \ket{\rho}_4 + \dfrac{1}{4} \ket{7}_4 \right) \otimes \ket{-} \right] = \left( \dfrac{\sqrt{15}}{4} \ket{\rho}_4 - \dfrac{1}{4} \ket{7}_4 \right) \otimes \ket{-} = \left(  \ket{\gamma}_4 - \dfrac{1}{2} \ket{7}_4 \right) \otimes \ket{-}
$$

On $\mathbb{STEP}$ 2.2, we apply the quantum gate $\boldsymbol{\varGamma}_4$ to the first 4 qubits of the computer. We use $\ket{\rho}_4$ to make more clear which part of the state has the index 7 in it and which one does not have it.
\begin{eqnarray*}
\boldsymbol{\varGamma}_4 \left( \ket{\gamma}_4 - \dfrac{1}{2} \ket{7}_4 \right) & = & \Big( 2 \ket{\gamma}_4 \bra{\gamma}_4 - \boldsymbol{I} \Big) \left(\ket{\gamma}_4 - \dfrac{1}{2} \ket{7}_4 \right) = \dfrac{3}{4} \ket{\gamma}_4 + \dfrac{1}{2} \ket{7}_4 = \dfrac{3 \sqrt{15}}{16} \ket{\rho}_4 + \dfrac{11}{16} \ket{7}_4 \\
\ket{\psi_3}_{4,1} & = & \left( \dfrac{3 \sqrt{15}}{16} \ket{\rho}_4 + \dfrac{11}{16} \ket{7}_4 \right) \otimes \ket{-}
\end{eqnarray*}

Thus, we have completed the first iteration of $\boldsymbol{G}_4$, and if we are to measure our quantum state now, we have a probability 
$$
p = \Big(\dfrac{11}{16} \Big)^2 = \dfrac{121}{256} \approx 0.4726
$$ 
of obtaining the index 7, and a probability 
$$
\bar{p} = \Big(\dfrac{3 \sqrt{15}}{16} \Big)^2 = \dfrac{135}{256} \approx 0.5274
$$ 
of obtaining any other index. As can be seen, the probability of finding the searched element is greater than of finding any other element, but is still not big enough, which tells us that at least another round of the algorithm is needed. \\

We repeat again the application of $\boldsymbol{G}_4$. First, we apply $\boldsymbol{O}_f$ and end up with the following quantum state, defined again as a combination of $\ket{\gamma}_4$ and $\ket{7}_4$.
$$
\ket{\psi_4}_{4,1} = \boldsymbol{O}_{f,4} (\ket{\psi_3}_{4,1})  = \left( \dfrac{3\sqrt{15}}{16} \ket{\rho}_4 - \dfrac{11}{16} \ket{7}_4 \right) \otimes \ket{-} = \left( \dfrac{3}{4} \ket{\gamma}_4 - \dfrac{14}{16} \ket{7}_4 \right) \otimes \ket{-}
$$

And then we apply $\boldsymbol{\varGamma}_4$:
$$
\ket{\psi_5}_{4} = \Big( 2 \ket{\gamma}_4 \bra{\gamma}_4 - \boldsymbol{I} \Big) \left( \dfrac{3}{4} \ket{\gamma}_4 - \dfrac{14}{16} \ket{7}_4 \right) = \dfrac{5}{16} \ket{\gamma}_4 + \dfrac{7}{8} \ket{7}_4 = \dfrac{5 \sqrt{15}}{64} \ket{\rho}_4 + \dfrac{61}{64} \ket{7}_4
$$

Thus, after two iterations of $\boldsymbol{G}_4$ we end up with a state whose probability of returning $7$ is
$$
p = \Big( \dfrac{61}{64} \Big)^2 = \dfrac{3721}{4096} \approx 0.9084,
$$ 
which gives us a fairly enough chance of completing the execution of Grover's search algorithm successfully. However, how can we be sure that we have obtained the maximum probability of retrieving the desired element from the database if we cannot measure the amplitudes inside our quantum register? Even more, whence comes the required order of $\mathcal{O}(\sqrt{2^n})$ needed in the number of iterations of $\boldsymbol{G}_4$? \\

Both questions are pivotal in the correct performance of the algorithm. Let us show their significance with a simple example: what happens if we continue applying the $\boldsymbol{G}_n$ gate to our quantum system? \\

If we apply $\boldsymbol{G}_4$ once more, we will have the following two states, the first after $\boldsymbol{O}_f$ and the second after $\boldsymbol{\varGamma}_4$.
$$
\ket{\psi_6} = \dfrac{5}{16} \ket{\gamma}_4 - \dfrac{33}{32} \ket{7}_4, \quad\quad
\ket{\psi_7} = - \dfrac{13 \sqrt{15}}{256} \ket{\rho}_4 - \dfrac{251}{256} \ket{7}_4
$$

As can be seen, the probability of obtaining index $7$ now is 
$$
p = \Big( \dfrac{251}{256} \Big)^2 \approx 0.9613.
$$ 

Yet, if we are still not satisfied with a $96 \%$ chance of success, we can run over $\boldsymbol{G}_4$ once again, and obtain
$$
\ket{\psi_8} = - \dfrac{13}{64} \ket{\gamma}_4 - \dfrac{238}{256} \ket{7}_4,\quad\quad 
\ket{\psi_9} = - \dfrac{342 \sqrt{15}}{2048} \ket{\rho}_4 + \dfrac{1562}{2048} \ket{7}_4,
$$ 
which gives us a probability $$p = \Big( \dfrac{1562}{2048} \Big)^2 \approx 0.5817$$ of obtaining $\ket{7}_4$. Yes, it seems that to unabatedly perform numberless iterations of $\boldsymbol{G}_n$ does not guarantee a continuous increment in the probability of success\footnote{{\em ``There is thy gold, worse poison to men's souls.''} (Romeo and Juliet, Act 5, Scene 1)}. Even more, it can waste all previously done work. We shall see next why this has happened.

\subsection*{Proof of correctness}

In this section we sketch a proof of the correctness of Grover's algorithm. The main ideas behind this proof are taken from \cite{boyer1996tight}, and will come in handy in the next subsection, where different and more general versions of the database search algorithm are discussed. 

\begin{thm}
Grover's search algorithm for a unique solution needs $m \sim \mathcal{O} \left( 2^{n/2} \right)$ iterations of $\boldsymbol{G}_n$ for maximizing the probability of obtaining the desired element with unknown index $j_0$.
\end{thm}

\begin{proof}

First, we are interested in redefining all the possible states that occur during the execution of the algorithm as a function of the different amplitudes involved. As was shown earlier, the only amplitude that will differ from the rest after every Grover's iteration is the one associated with the basis state that identifies the searched element. Thus, it suffices to define the generic state in the following way: 
$$
\ket{\psi (\alpha,\beta)}_n = \alpha \ket{j_0}_n + \beta \sum_{\substack{j=0 \\ j\neq j_0}}^{2^n-1} \ket{j}_n,
$$ 
with $\alpha$ and $\beta$ real numbers constrained to $\alpha^2 + (2^n - 1) \beta^2 = 1$.

Note that we are only having into account the first register, the one with $n$ qubits. The remaining qubit will behave as explained before, starting at $\ket{1}$ and remaining as $\ket{-}$ throughout the rest of the algorithm, but for the sake of simplicity will be omitted during this proof. \\

Let us see what happens when we apply all quantum gates that make up the Grover transform $\boldsymbol{G}_n$ to a generic state $\ket{\psi(\alpha,\beta)}_n$. First, we perform $\boldsymbol{O}_f$, which recognizes the searched element and flips its amplitude. After $\boldsymbol{O}_f$ we obtain
$$
\boldsymbol{O}_f \ket{\psi(\alpha,\beta)}_n = - \alpha \ket{j_0}_n + \beta \sum_{\substack{j=0 \\ j\neq j_0}}^{2^n-1} \ket{j}_n
$$

In order to apply the $\boldsymbol{\varGamma}_n$ quantum gate, it is interesting to previously see our current state as a function of $\ket{\gamma}_n$, as was done before:
$$
\boldsymbol{O}_f \ket{\psi(\alpha,\beta)}_n = - (\alpha + \beta) \ket{j_0}_n + \sqrt{2^n} \beta \ket{\gamma}_n.
$$ 

We are finally ready to apply $\boldsymbol{\varGamma}_n$:
\begin{eqnarray*}
\boldsymbol{G}_n \ket{\psi(\alpha,\beta)}_n  & = & \boldsymbol{\varGamma}_n \boldsymbol{O}_f \ket{\psi(\alpha,\beta)}_n \\ 
& = & \boldsymbol{\varGamma}_n \Big( \ket{\psi'(\alpha,\beta)}_n \Big) = \boldsymbol{\varGamma}_n \Big(-(\alpha + \beta) \ket{j_0}_n + \sqrt{2^n} \beta \ket{\gamma}_n \Big) \\
& = & \Big( 2 \ket{\gamma}_n \bra{\gamma}_n - \boldsymbol{I}^{\otimes n} \Big) \Big(-(\alpha + \beta) \ket{j_0}_n + \sqrt{2^n} \beta \ket{\gamma}_n \Big) \\
& = &( \alpha + \beta ) \ket{j_0}_n + \Big( 2 \sqrt{2^n} \beta - \dfrac{2}{\sqrt{2^n}} (\alpha + \beta) - \sqrt{2^n} \beta \Big)  \ket{\gamma}_n\\
& = & \Big( \dfrac{2^{n-1} - 1}{2^{n-1}} \alpha  + \dfrac{2^n - 1}{2^{n-1}} \beta \Big) \ket{j_0}_n + \Big( - \dfrac{1}{2^{n-1}} \alpha + \dfrac{2^{n-1} - 1}{2^{n-1}} \beta  \Big) \sum_{\substack{j=0 \\ j\neq j_0}}^{2^n-1} \ket{j}_n.
\end{eqnarray*}

Now that we know the effect $\boldsymbol{G}_n$ makes to a generic state, we are in a position to predict in which iteration of the algorithm we have more possibilities of obtaining the desired element $j_0$. If we define 
$$ 
\ket{\psi_{k+1}(\alpha_{k+1}, \beta_{k+1})}_n = \boldsymbol{G}_n (\ket{\psi_k(\alpha_k,\beta_k)}_n),
$$ 
where $\alpha_1 = \beta_1 = 1/\sqrt{2^n}$, and 
$$
\alpha_{k+1} = \dfrac{2^{n-1} - 1}{2^{n-1}} \alpha_k + \dfrac{2^n - 1}{2^{n-1}} \beta_k, \quad \quad
\beta_{k+1} = - \dfrac{1}{2^{n-1}} \alpha_k + \dfrac{2^{n-1} - 1}{2^{n-1}} \beta_k
$$
for $j \geq 1$, then we can try to find a more tractable closed-form formula for the amplitude of $i_0$. Note that, for $j = 0$, it is not possible to define $\ket{\psi_0}$ as a function of $\alpha_0$ and $\beta_0$, thus only the cases where $j \geq 1$ will be defined as such. \\

If we designate $\theta \in [0,2\pi)$ such that $\sin^2 \theta = 1/2^n$, we can easily prove inductively that  
$$
\alpha_j = \sin \Big( (2j-1)\theta \Big), \qquad \beta_j = \dfrac{1}{\sqrt{2^n-1}} \cos \Big( (2j-1)\theta \Big).
$$

Let us suppose now that, for an unknown step $j = m + 1$ (note that $m$ is equivalent to the number of times we have applied $\boldsymbol{G}_n$), we want to assure that $\alpha_m = 1$. This occurs when $(2m+1)\theta = \pi/2$, and expressly, when $m = (\pi - 2 \theta)/(4 \theta)$.

Obviously, we can not perform a non-integer number of iterations of $\boldsymbol{G}_n$. If we take $m = \left\lfloor \pi/(4 \theta) \right\rfloor$, we can conclude that the number of iterations of $\boldsymbol{G}_n$ needed for achieving the maximum probability of success is close $(\pi/4) \sqrt{2^n}$ (note that $\theta \approx \sin \theta = 1/\sqrt{2^n}$ when $2^n$ is large enough). 

Thus, we can conclude that the number $m$ of iterations of $\boldsymbol{G}_n$ has order $m \sim \mathcal{O} \big( \sqrt{2^n} \big)$. 
\end{proof}

\subsection*{Multiple solutions}

The main limitation of Grover's search algorithm deals with the number of solutions: it assumes that there is one and only one element in the database that matches our search. There are many cases in which Grover's search may prove useful, but in which the number of solutions is unknown, or maybe we do not even know if there is indeed a solution. In this subsection we proceed to explain an alternate version of Grover's algorithm that originally appeared in \cite{boyer1996tight} which takes care of this drawback. \\

Let us suppose that we have an unstructured database, indexed by $0,1,\ldots,2^n-1$. We are interested in finding an element inside the database that fulfills a certain property, but we do not know if such an element exists, or if there is more than one. We name $A \subseteq \{0, 1, \ldots , 2^n - 1\}$ the set of possible solutions, with $\#A = t \in \{0, 1, \ldots , 2^n - 1\}$. At first sight, one can only hope to just obtain the same performance results as in the original Grover's algorithm just by applying $\boldsymbol{G}_n$ the same number of steps. However, it is easy to find a counterexample showing that the probability of success after $(\pi/4)\sqrt{2^n}$ iterations changes dramatically when $t \neq 1$. \\

In order to show how this variation of the algorithm works, we define $B = \{0,\ldots,2^n-1\} \setminus A$, with $\# B = 2^n - t$. Following a similar approach, we can assume that every quantum state of our system after first applying globally the Hadamard transform can be expressed as
$$
\ket{\psi(\alpha,\beta)}_n = \alpha \sum_{i \in A} \ket{i}_n + \beta \sum_{i \in B} \ket{i}_n, \qquad \text{ with } \alpha, \beta \in [0,1],
$$ 
where $t \alpha^2 + (2^n - t) \beta^2 = 1$. 

Note that we only take into account the first $n$ qubits, as previously. We would like to clarify that the algorithm is essentially the same, so there is no need for explaining it again. The only substantial change is the expected number of iterations of $\boldsymbol{G}_n$ needed for maximizing the probability of success. Thus, if we apply $\boldsymbol{O}_n$ to a certain state, we end up with the following configuration:
$$
\boldsymbol{O}_n \ket{\psi(\alpha,\beta)}_n = - \alpha \sum_{i \in A} \ket{i}_n + \beta \sum_{i \in B} \ket{i}_n
$$

Just like before, after a complete iteration of $\boldsymbol{G}_n$ the quantum system is in the state:
$$
\boldsymbol{\varGamma}_n \boldsymbol{O}_n \Big( \ket{\psi(\alpha,\beta)}_n \Big) = \Big( \dfrac{2^{n-1} - t}{2^{n-1}} \alpha  + \dfrac{2^n - t}{2^{n-1}} \beta \Big) \sum_{i \in A} \ket{i}_n + \Big( - \dfrac{t}{2^{n-1}} \alpha_j + \dfrac{2^{n-1} - t}{2^{n-1}} \beta_j  \Big) \sum_{i \in B} \ket{i}_n
$$

As done previously, if we define 
$$
\boldsymbol{G}_n (\ket{\psi_j(\alpha_j,\beta_j)}_n) = \ket{\psi_{j+1}(\alpha_{j+1}, \beta_{j+1})}_n,
$$ 
we can induce a recursive formula for the general state of the system, where $\alpha_1 = \beta_1 = 1/ \sqrt{2^n}$ and 
$$
\alpha_{j+1} = \dfrac{2^{n-1} - t}{2^{n-1}} \alpha_j + \dfrac{2^n - t}{2^{n-1}} \beta_j, \quad\quad 
\beta_{j+1} = - \dfrac{t}{2^{n-1}} \alpha_j + \dfrac{2^{n-1} - t}{2^{n-1}} \beta_j
 $$

If we define $\theta \in [0,2\pi)$ such that $\sin^2 \theta = t/2^n$,then we can arrive at the following closed-up formula, just by using induction and some trigonometric identities:
$$
\alpha_j = \dfrac{1}{\sqrt{t}} \sin ((2j-1)\theta), \quad\quad 
\beta_j = \dfrac{1}{\sqrt{2^n-t}} \cos ((2j-1)\theta).
$$

\begin{thm}
Let $t$ and $\theta$ be as previously defined. Let $m$ be an arbitrary positive integer and $j$ be an integer chosen at random from $\{1,\ldots,m\}$. Then, if we apply $(j-1)$ times the gate $\boldsymbol{G}_n$ to the state $\ket{\gamma}_n$, followed by a measurement of the state, the average probability of obtaining one of the $t$ solutions in $A$ is 
$$
P_m = \dfrac{1}{2} - \dfrac{\sin (4m\theta)}{4m \sin (2\theta)}.
$$
\end{thm}

\begin{proof}
After $j-1$ iterations of $\boldsymbol{G}_n$, the probability of obtaining one of the possible solutions is $$
t \alpha_j^2 = \sin^2 ((2j-1) \theta).
$$

If $1 \leq j \leq m$ is chosen randomly, then the average probability of success is given by
$$
P_m = \sum_{j = 1}^{m} \dfrac{1}{m} \sin^2 ((2j - 1)\theta) = \dfrac{1}{2m} \sum_{j=1}^{m} \Big( 1 - \cos ((2j - 1) 2\theta ) \Big) = \dfrac{1}{2} - \dfrac{\sin (4m\theta)}{4m \sin (2\theta)}.
 $$

In the last step, we have used the following trigonometric identity:
$$
\sum_{j = 1}^{m} \cos \Big( (2j - 1) \phi \Big) = \dfrac{\sin (2m \phi)}{2 \sin \phi}, \quad \forall \phi \in [0,2 \pi).
$$
\end{proof}

\begin{cor}
    Under the previous assumptions, if $m \geq 1/\sin2\theta$, then $P_m \geq 1/4$.
\end{cor}

Thus, Grover's search algorithm for multiple solutions follows the next scheme:
$$$$\noindent\framebox{
	\parbox[b]{\linewidth}{
		\begin{algorithmic}
			\State $m \gets 1$
			\State $\lambda \gets \lambda \in (1,4/3)$
			\State $j \gets \text{random}\{1, \ldots, m \}$
			\State $\ket{\psi} \gets \boldsymbol{G}_n^{j-1} \ket{\gamma}_n \quad$ \#Apply the $\boldsymbol{G}_n$ quantum gate $j-1$ times to $\ket{\gamma}_n$
            \State $i \gets$ Value of the first register
			\If {$i \in A$}
		          \State return $i$
			\Else
                \State $m \gets \min \big(\lambda m, \sqrt{2^n} \big)$
				\State go to line 3
		      \EndIf
		\end{algorithmic}
	}
} \\

Let us explain a little the inner workings of this variation. First, we set $m$ as an upper bound for the number of $\boldsymbol{G}_n$ performances. Then, we perform $j-1<m$ such operations. If we do not obtain the desired result, we try again with a slightly bigger $m$.

The average number of iterations of Grover's algorithm is given via the following result:

\begin{thm}
Let $t = \# A \leq 2^n \cdot 3/4$, the expected number of iterations (that is, of computations involving $\boldsymbol{G}_n^{j-1}$) for finding a solution with the previous algorithm has order $\mathcal{O} (\sqrt{2^n/t})$.
\end{thm}

\begin{proof} \cite{boyer1996tight} Let 
$$
m_0 = \dfrac{1}{\sin (2 \theta)} = \dfrac{1}{2\sin (\theta) \cos (\theta)} = \dfrac{2^{n-1}}{\sqrt{(2^n - t)t}} < \sqrt{\dfrac{2^n}{t}}, \qquad \text{as } t\leq 2^n \cdot \frac{3}{4}.
$$ 

We call $m_0$ the critical stage or, more generally, the first $m$ such that $m \geq m_0$. According to our corollary, from that moment on we have a success probability greater than $1/4$, although the algorithm might have stopped earlier, if we are lucky.

The expected total number of iterations that we need to reach the critical stage is at most 
$$
\dfrac{1}{2} \sum_{s=1}^{\Ceil{\log_{\lambda} m_0}} \lambda^{s-1} < \dfrac{1}{2} \dfrac{\lambda}{\lambda-1} m_0 = 3 m_0
$$ 
and the expected number of Grover iterations needed to succeed once the critical point has been reached is bounded by
$$
\dfrac{1}{2} \sum_{u=0}^{\infty} \left( \dfrac{3}{4} \right)^u \lambda^{u + \Ceil{\log_{\lambda} m_0}} < \frac{2\lambda}{4 - 3\lambda} m_0 = 4 m_0.
$$ 

Thus, the expected number of Grover iterations is upper bounded by $8m_0 \sim \mathcal{O} (\sqrt{2^n/t})$.
\end{proof}

Grover's algorithm, which is also known as quantum amplitude amplification, is specially useful in problems where the best known classical solution is to iterate over all possible candidates in the worst case. Although the acceleration is not exponential, but quadratic, we do know that it is a strict improvement with respect to the classical approach, provided that we are able to identify a solution in polynomial time. The superiority of Shor's algorithm, on the other hand, relies on the unproven conjecture that factoring is not in P. An adiabatic version of Grover's algorithm can be found in \cite{farhi2000quantum} and \cite{roland2002quantum}.

\newpage

\section{Quantum Counting}\label{sec:counting}

If we want to apply Grover's algorithm without knowing beforehand the number of solutions (or, at the very least, a good estimation) it is clear we might be in trouble as the probability of succeeding can be dramatically altered (for the worst) if we perform too many Grover gates.

A product of the work of Gilles Brassard, Peter Høyer and Alain Tapp, the quantum counting algorithm is a variation of the Grover's search in which we use the quantum Fourier transform for counting the solutions to the database search problem, in case that this number is unknown to us. It was first sketched in \cite{boyer1996tight} and thoroughly described in \cite{brassard1998quantum}. Additionally, some details of the proof of its correctness that we present here are taken from \cite{diao2012quantum}. \\

The counting problem can be seen as the following: let us suppose that we have an unstructured database, whose elements are indexed by $\{0,1,\ldots,2^n-1\}$, and that we want to know how many elements fulfill a certain property. Note that we are interested in the number of solutions, not in returning any of them. If $A \subseteq \{0,\ldots,2^n-1\}$ is the set of indices that fulfill our query, the quantum counting problem can be seen as the problem of calculating $t = \#A$. We shall also define $B = \{0,\ldots,2^n-1\} \setminus A$ as the set of indices that do not fulfill the property. Therefore, $\#B = 2^n -t$. We will assume, quite reasonably, $t \leq 2^{n-1}$.

\paragraph{$\mathbb{SETUP}$}
$$$$\noindent\framebox{\parbox[b]{\linewidth}{\begin{algorithmic}
\State $\ket{\psi_0}_{p,n} \leftarrow \ket{0}_p \otimes \ket{0}_n$
\end{algorithmic}}} \\

The quantum counting algorithm starts with $p+n$ qubits initialized at $0$, where the value $p$ will be dealt with later.

\paragraph{$\mathbb{STEP}$ 1}
$$$$\noindent\framebox{\parbox[b]{\linewidth}{\begin{algorithmic}
\State $\ket{\psi_1}_{p,n} \leftarrow  \boldsymbol{H}^{\otimes p+n} \left( \ket{\psi_0}_{p,n} \right)$
\end{algorithmic}}} \\

As usual, the first step of the algorithm involves the Hadamard gate, which gives us the state
$$
\ket{\psi_1}_{p,n} = \boldsymbol{H}^{\otimes p+n} \left( \ket{\psi_0}_{p,n} \right) = \ket{\gamma}_p \otimes \ket{\gamma}_n,
$$
where we recall that 
$$
\ket{\gamma}_n =  \dfrac{1}{\sqrt{2^n}} \sum_{k=0}^{2^n-1} \ket{k}_n.
$$

If we use the notations 
\begin{eqnarray*}
\ket{a}_n = \dfrac{1}{\sqrt{t}} \sum_{k \in A} \ket{k}_n, & &
\ket{b}_n = \dfrac{1}{\sqrt{2^n-t}} \sum_{k \in B} \ket{k}_n, \\
\ket{\mu^{+}}_n = \dfrac{1}{\sqrt{2}} \left( \ket{b}_n - i \ket{a}_n \right), & & 
\ket{\mu^{-}}_n = \dfrac{1}{\sqrt{2}} \left( \ket{b}_n + i \ket{a}_n \right)
\end{eqnarray*}
and define $\omega \in [0,1]$ such that $\sin^2 (\pi \omega) = t / 2^n$, we can express $\ket{\gamma}_n$ as:
\begin{eqnarray*}
\ket{\gamma}_{n} & = & \dfrac{1}{\sqrt{2^n}} \sum_{k=0}^{2^n-1} \ket{k}_n = \dfrac{1}{\sqrt{2^n}} \left( \sum_{k \in A} \ket{k}_n + \sum_{k \in B} \ket{k}_n \right) \\
& = & \sin (\pi \omega) \left( \dfrac{1}{\sqrt{t}} \sum_{k \in A} \ket{k}_n \right) + \cos (\pi \omega) \left( \dfrac{1}{\sqrt{2^n-t}} \sum_{k \in B} \ket{k}_n \right) = \sin (\pi \omega) \ket{a}_n + \cos (\pi \omega) \ket{b}_n \\
& = & \dfrac{i \sin (\pi \omega)}{\sqrt{2}} \Big( \ket{\mu^+}_n - \ket{\mu^-}_n \Big) + \dfrac{\cos (\pi \omega)}{\sqrt{2}} \Big( \ket{\mu^+}_n + \ket{\mu^-}_n \Big) = \dfrac{e^{\pi i \cdot\omega}}{\sqrt{2}} \ket{\mu^+}_n + \dfrac{e^{-\pi i \cdot\omega}}{\sqrt{2}} \ket{\mu^-}_n
\end{eqnarray*}

From now on, $\omega$ will be our target instead of $t$, as $t$ and $\omega$ define each other. Note that $t \leq 2^{n-1}$ translates into $\omega \leq 1/2$.

\paragraph{$\mathbb{STEP}$ 2}
$$$$\noindent\framebox{\parbox[b]{\linewidth}{\begin{algorithmic}
\State $\ket{\psi_2}_{p,n} \leftarrow \boldsymbol{C}_{p,n} (\ket{\psi_1}_{p,n})$
\end{algorithmic}}} \\

For the next step, we need a new unitary transformation called the counting gate, which is represented by $\boldsymbol{C}_{p,n}$. This quantum gate has the following effect on a quantum state of the form $\ket{m}_p \otimes \ket{\psi}_n $, where $\ket{m}_p$ is a basis state on $p$ qubits, $\ket{\psi}_n$ is any quantum state on $n$ qubits and $\boldsymbol{G}_n$ is the Grover gate described in the previous algorithm:
$$
\boldsymbol{C}_{p,n} : \ket{m}_p \otimes \ket{\psi}_n  \longmapsto \ket{m}_p \otimes (\boldsymbol{G}_n)^m \ket{\psi}_n
$$

Let us see what happens when we apply this new quantum gate to our previous quantum system.
\begin{equation*}
\begin{split}
\ket{\psi_2}_{p,n} & = \boldsymbol{C}_{p,n} \left( \ket{\gamma}_p \otimes \ket{\gamma}_n \right) = \dfrac{1}{\sqrt{2^p}} \sum_{j=0}^{2^p-1} \left[ \ket{j}_p \otimes \boldsymbol{G}_{n}^{j} \left( \ket{\gamma}_n \right) \right]  \\
 & = \dfrac{1}{\sqrt{2^p}} \sum_{j=0}^{2^p-1} \left[ \ket{j}_p \otimes \boldsymbol{G}_{n}^{j} \left( \dfrac{e^{\pi i \cdot\omega}}{\sqrt{2}} \ket{\mu^+}_n + \dfrac{e^{-\pi i \cdot\omega}}{\sqrt{2}} \ket{\mu^-}_n \right)  \right] \\
 & = \dfrac{1}{\sqrt{2^{p+1}}} \sum_{j=0}^{2^p-1} \bigg[ \ket{j}_p \otimes \bigg( e^{\pi i \cdot\omega} \boldsymbol{G}_{n}^{j} \left( \ket{\mu^+}_n \right) + e^{-\pi i \cdot\omega} \boldsymbol{G}_{n}^{j} \left( \ket{\mu^-}_n \right) \bigg)  \bigg] \\
 & = \dfrac{1}{\sqrt{2^{p+1}}} \sum_{j=0}^{2^p-1} \bigg[ \ket{j}_p \otimes \bigg( e^{\pi i \cdot\omega(2j+1)} \ket{\mu^+}_n  + e^{-\pi i \cdot\omega(2j+1)} \ket{\mu^-}_n \bigg) \bigg] \\
 & =  \dfrac{e^{\pi i \cdot\omega}}{\sqrt{2^{p+1}}} \sum_{j=0}^{2^p-1} e^{2\pi i \cdot \omega j} \ket{j}_p \otimes \ket{\mu^+}_n + \dfrac{e^{-\pi i \cdot\omega}}{\sqrt{2^{p+1}}} \sum_{j=0}^{2^p-1} e^{2\pi i \cdot (1- \omega) j} \ket{j}_p \otimes \ket{\mu^-}_n
\end{split}
\end{equation*} \\

In order to understand what happened in the last identities of the equation, we should see the effect that $\boldsymbol{G}_{n}^{j}$ has on the states $\ket{\mu^+}_n$ and $\ket{\mu^-}_n$. First:

\begin{eqnarray*}
\boldsymbol{G}_{n} \left( \ket{a}_n \right) & = & \Big( 2 \ket{\gamma}_n \bra{\gamma}_n - \boldsymbol{I} \Big) \big( - \ket{a}_n \big) = - 2 \ket{\gamma}_n \braket{\gamma | a}_n + \ket{a}_n = - 2 \sin (\pi\omega) \ket{\gamma}_n + \ket{a}_n \\
& = & -2 \sin (\pi\omega) \big(\cos(\pi\omega)\ket{b}_n + \sin(\pi\omega) \ket{a}_n\big) + \ket{a}_n \\ 
& = & - 2 \sin (\pi\omega) \cos (\pi\omega) \ket{b}_n + (1-2\sin^2(\pi\omega)) \ket{a}_n = - \sin (2\pi\omega) \ket{b}_n + \cos (2\pi\omega) \ket{a}_n
\end{eqnarray*}
\begin{eqnarray*}
\boldsymbol{G}_{n} \left( \ket{b}_n \right) & = & \Big( 2 \ket{\gamma} \bra{\gamma}_n) - \boldsymbol{I} \Big) \big(\ket{b}_n \big) = 2 \ket{\gamma} \braket{\gamma | b}_n - \ket{b}_n = 2 \cos (\pi\omega) \ket{\gamma}_n - \ket{b}_n \\
& = & 2 \cos (\pi\omega) \big(\cos(\pi\omega)\ket{b}_n + \sin(\pi\omega) \ket{a}_n\big) - \ket{b}_n \\
& = & (2\cos^2(\pi\omega) - 1) \ket{b}_n + 2 \sin (\pi\omega) \cos (\pi\omega) \ket{a}_n = \cos (2\pi\omega) \ket{b}_n + \sin (2\pi\omega) \ket{a}_n
\end{eqnarray*}
which leads us to
\begin{eqnarray*}
\boldsymbol{G}_{n} \left( \ket{\mu^+}_n \right) & = & \dfrac{1}{\sqrt{2}} \Big( \cos (2\pi\omega) \ket{b}_n + \sin (2\pi\omega) \ket{a}_n + i \sin (2\pi\omega) \ket{b}_n - i \cos (2\pi\omega) \ket{a}_n \Big) \\
& = & \dfrac{e^{2\pi i \cdot\omega}}{\sqrt{2}} (\ket{b}_n - i \ket{a}_n)= e^{2\pi i \cdot\omega} \ket{\mu^+}_n
\end{eqnarray*}
\begin{eqnarray*}
\boldsymbol{G}_{n} \left( \ket{\mu^-}_n \right) & = & \dfrac{1}{\sqrt{2}} \Big( \cos (2\pi\omega) \ket{b}_n + \sin (2\pi\omega) \ket{a}_n - i \sin (2\pi\omega) \ket{b}_n + i \cos (2\pi\omega) \ket{a}_n \Big) \\
& = & \dfrac{e^{-2\pi i \cdot\omega}}{\sqrt{2}} (\ket{b}_n + i \ket{a}_n) = e^{-2\pi i \cdot\omega} \ket{\mu^-}_n
\end{eqnarray*}

After that, instead of writing $e^{-2\pi i \cdot\omega j}$ in the coefficients of $\ket{j}_p \otimes \ket{\mu^-}_n$ we opt to write $e^{2\pi i \cdot (1-\omega) j}$ (clearly the same coefficient, after all) so we keep the angular parameter in $[0,1]$.

\paragraph{$\mathbb{STEP}$ 3}
$$$$\noindent\framebox{\parbox[b]{\linewidth}{\begin{algorithmic}
\State $\ket{\psi_3}_{p,n} \leftarrow ( \boldsymbol{F}^{-1}_{p} \otimes \boldsymbol{I}^{\otimes n}) (\ket{\psi_2}_{p,n})$
\end{algorithmic}}} \\

The next step involves $\boldsymbol{F}_p^{-1}$, the inverse of the quantum Fourier transform, already explained in Section \ref{sec::shor}, which has the following effect on an $p$-qubit basis state.
$$
\boldsymbol{F}^{-1}_{p} : \ket{j}_p \longmapsto \dfrac{1}{\sqrt{2^p}} \sum_{l=0}^{2^p-1} e^{- 2 \pi i \cdot jl/2^p} \ket{l}_p
$$

We recall that, applied to a certain type of quantum state, the inverse quantum Fourier transform can give us certain information about the distribution of its amplitudes.
$$
\boldsymbol{F}^{-1}_{p} \left( \dfrac{1}{\sqrt{2^p}} \sum_{j=0}^{2^p-1} e^{2 \pi i \cdot kj/2^p} \ket{j}_p \right) = \boldsymbol{F}^{-1}_{p} \Big( \boldsymbol{F}_{p} \ket{k}_p \Big) = \ket{k}_p.
$$

Let us see what happens when we apply it to our current quantum state:
\begin{eqnarray*}
\ket{\psi_3}_{p,n} & = & \Big( \boldsymbol{F}^{-1}_{p} \otimes \boldsymbol{I}^{\otimes n} \Big) \left( \ket{\psi_2} \right)_{p,n} \\
& = & \dfrac{e^{\pi i \cdot\omega}}{\sqrt{2^{p+1}}} \sum_{j=0}^{2^p-1} e^{2\pi i \cdot \omega j} \boldsymbol{F}^{-1}_{p}  \ket{j}_p  \otimes \ket{\mu^+}_n + \dfrac{e^{-\pi i \cdot\omega}}{\sqrt{2^{p+1}}} \sum_{j=0}^{2^p-1} e^{2\pi i \cdot (1- \omega) j} \boldsymbol{F}^{-1}_{p}  \ket{j}_p  \otimes \ket{\mu^-}_n \\
& = & \dfrac{e^{\pi i \cdot\omega}}{2^p\sqrt{2}} \sum_{l=0}^{2^p-1} \sum_{j=0}^{2^p-1} e^{2\pi i \cdot j \left(\omega - l/2^p \right)} \ket{l}_p \otimes \ket{\mu^+}_n + \dfrac{-e^{\pi i \cdot\omega}}{2^p\sqrt{2}} \sum_{l=0}^{2^p-1} \sum_{j=0}^{2^p-1} e^{2\pi i \cdot j \left( (1-\omega) - l/2^p \right)} \ket{l}_p \otimes \ket{\mu^-}_n
\end{eqnarray*}

Now mind the following remark: if we can write $\omega = h/2^p$, with $h \in \{0,\ldots,2^p\}$, then the above property of $\boldsymbol{F}^{-1}_p$ can be applied and we would have
\begin{eqnarray*}
\ket{\psi_3}_{p,n} & = & 
\dfrac{e^{\pi i \cdot\omega}}{\sqrt{2}} \boldsymbol{F}^{-1}_{p} \left( \dfrac{1}{\sqrt{2^p}} \sum_{j=0}^{2^p-1} e^{2\pi i \cdot hj/2^p} \ket{j}_p \right) \otimes \ket{\mu^+}_n + \\
&& \qquad \qquad + \dfrac{e^{-\pi i \cdot\omega}}{\sqrt{2}} \boldsymbol{F}^{-1}_{p} \left( \dfrac{1}{\sqrt{2^p}} \sum_{j=0}^{2^p-1} e^{2\pi i \cdot (2^p-h)j/2^p}  \ket{j}_p \right) \otimes \ket{\mu^-}_n\\
& = & \dfrac{e^{\pi i \cdot\omega}}{\sqrt{2}} \ket{h}_p \otimes \ket{\mu^+}_n + \dfrac{e^{-\pi i \cdot\omega}}{\sqrt{2}} \ket{2^p-h}_p \otimes \ket{\mu^-}_n
\end{eqnarray*}

Here the value of $p$ appears finally as a precision choice. The greater $p$ is, the more possibilities we have to have an expression $\omega = h/2^p$, for $h \in \mathbb{Z}$. Or, at the very least, the more accurate an approximation of this kind will be (more on that later).\\

Anyway, assuming $\omega = h/2^p$ is necessary for the last expression of $\ket{\psi_3}_{p,n}$ to happen, so we might not have such a closed formula, but nevertheless it may help to understand why the last step is the following:

\paragraph{$\mathbb{STEP}$ 4}
$$$$\noindent\framebox{\parbox[b]{\linewidth}{\begin{algorithmic}
	\State $\tilde l \leftarrow$ measure the first register of $\ket{\psi_3}_{p,n}$
	\If {$\tilde l > 2^{p-1}$}
		\State $\tilde l \leftarrow 2^p - \tilde l$
	\EndIf
	\State $\tilde t \leftarrow 2^n \sin^2 \left( \dfrac{\pi \tilde l}{2^p} \right)$
    \State return $\tilde t$
\end{algorithmic}}} \\

First, note that $\omega = h/2^p$ and $\omega \leq 1/2$ translates into $h \leq 2^{p-1}$ (in case such an $h$ exists). So, our \textbf{if} step implies that, whether we get $h = 2^p \omega$ or $2^p - h = 2^p(1- \omega)$ from the measurement, the value $\tilde{l}$ corresponds to $h$.\\

However, as that might not be the case, we end up with $\tilde l = 2^p \tilde\omega$ or $\tilde l = 2^p(1- \tilde\omega)$, where $\tilde\omega$ is an estimation of $\omega$, and from which we can get an estimation $\tilde t$ for $t$. Thus ends the quantum counting algorithm, it remains to be proven how good is that estimator, and the probability of obtaining it.

\subsection*{Proof of correctness (and accuracy)}

\begin{thm}
Let $f : \{0,1,\ldots,2^n-1\} \rightarrow \{0,1\}$ be the indicator function of a certain set $A \subseteq \{0,1,\ldots,2^n-1\}$ with $t = |A| = |f^{-1}(1)| \leq 2^{n-1}$, let $p \geq 2$ and let $\tilde t$ be the return value of the previously described algorithm. 

Then, with probability of at least $8 / \pi^2$, 
$$
|t - \tilde t| \leq \dfrac{2\pi}{2^p} \sqrt{t( 2^n-t)} + \dfrac{\pi^2}{2^{2p}}\abs{2^n - t}.
$$
\end{thm}

\begin{proof}
\cite{boyer1996tight, brassard1998quantum} Let us begin with our last unconditional expression:
$$
\ket{\psi_3}_{p,n} = \dfrac{e^{\pi i \cdot\omega}}{2^p\sqrt{2}} \sum_{l=0}^{2^p-1} \sum_{j=0}^{2^p-1} e^{2\pi i \cdot j \left(\omega - l/2^p \right)} \ket{l}_p \otimes \ket{\mu^+}_n + \dfrac{-e^{\pi i \cdot\omega}}{2^p\sqrt{2}} \sum_{l=0}^{2^p-1} \sum_{j=0}^{2^p-1} e^{2\pi i \cdot j \left( (1-\omega) - l/2^p \right)} \ket{l}_p \otimes \ket{\mu^-}_n
$$

As noted previously, the probability of obtaining a certain value $\tilde l$ in Step 4 is the combining probability of obtaining $\tilde l = 2^p \tilde\omega$ or $\tilde l = 2^p(1- \tilde\omega)$ so we would need to consider the next amplitudes
$$
\def\arraystretch{3}
\left\{ \begin{array}{lcll}
\alpha_{\tilde l} &=& \displaystyle \dfrac{e^{\pi i \cdot\omega}}{2^p\sqrt{2}}  \sum_{j=0}^{2^p-1} e^{2\pi i \cdot j \left(\omega - \tilde l/2^p \right)} & \text{ for } \tilde l \leq 2^{p-1} \\
\alpha_{\tilde l} &=& \displaystyle \dfrac{-e^{\pi i \cdot\omega}}{2^p\sqrt{2}} \sum_{j=0}^{2^p-1} e^{2\pi i \cdot j \left((1-\omega) - \tilde l/2^p \right)} & \text{ for } \tilde l > 2^{p-1}  \\
\end{array} \right.
$$ 

We will focus on the amplitude and probability coming from $\tilde l = 2^p \tilde\omega$, as the other is analogous. We can write
$$
\sum_{j=0}^{2^p-1} e^{2\pi i \cdot j\left(\omega - \tilde l/2^p \right)} = \frac{1- e^{2\pi i \cdot \left(\omega - \tilde l/2^p\right)2^p}}{1 - e^{2\pi i \cdot \left(\omega - \tilde l/2^p\right)}},
$$
and note $\left| 1 - e^{i \alpha} \right|^2 = 4 \sin^2 (\alpha/2)$, so
$$
\left| \sum_{j=0}^{2^p-1} e^{2\pi i \cdot j\left(\omega - \tilde l/2^p \right)} \right|^2 = \frac{\sin^2 \left(2^p \pi (\omega- \tilde l/2^p) \right)}{\sin^2 \left(\pi (\omega- \tilde l/2^p) \right)}.
$$

So, the probability of obtaining $\tilde l = 2^p \tilde \omega$ is given by
$$
\dfrac{1}{2^{2p+1}} \cdot \frac{\sin^2 \left(2^p \pi (\omega- \tilde l/2^p) \right)}{\sin^2 \left(\pi (\omega- \tilde l/2^p) \right)}.
$$ \\

Let us write now
$$
l_1 = \lfloor 2^p\omega \rfloor, \qquad l_2 = \lceil 2^p\omega \rceil, \qquad \Delta = \omega - \frac{l_1}{2^p},
$$ 
which are, respectively, the closest values of $\tilde l$ to $2^p \omega$ and an estimate of the relative error. We will assume (the other case is analogous) that $2^p \omega$ is actually closer to $l_1$, that is, the error is bounded by $0 \leq \Delta < 1/2^{p+1}$ and we also have 
$$
\frac{l_2}{2^p} = \omega + \left( \frac{1}{2^p}-\Delta \right), \qquad \text{ with } \left( \frac{1}{2^p}-\Delta \right) \geq \frac{1}{2^{p+1}}.
$$

We can then write
\begin{eqnarray*}
P\left(\left|\dfrac{\tilde l}{2^p}-\omega\right| \leq \dfrac{1}{2^p} \right) & = & 
P \left( |\tilde l-2^p\omega|\leq 1 \right) = P(\tilde l = l_1) + P(\tilde l=l_2) = |\alpha_{l_1}|^2 + |\alpha_{l_2}|^2 \\
& = & \displaystyle \dfrac{1}{2^{2p+1}} \cdot \frac{\sin^2 \big(2^p \pi (\omega- l_1/2^p) \big)}{\sin^2 \big(\pi (\omega- l_1/2^p) \big)} + \dfrac{1}{2^{2p+1}} \cdot \frac{\sin^2 \big( 2^p \pi (\omega- l_2/2^p) \big)}{\sin^2 \big( \pi (\omega- l_2/2^p) \big)} \\
& = & \dfrac{1}{2^{2p+1}} \cdot \dfrac{\sin^2(2^p\pi\Delta)}{\sin^2(\pi\Delta)} + \dfrac{1}{2^{2p+1}} \cdot \dfrac{\sin^2 \big( 2^p\pi(\Delta - 1/2^p) \big)}{\sin^2 \big( \pi(\Delta - 1/2^p) \big)} \\
& \geq & \frac{1}{2^{2p+1}} \left( \frac{1}{\sin^2(\pi\Delta)} + \frac{1}{\sin^2 \big( \pi(\Delta - 1/2^p)} \right) \\
& \geq & \frac{1}{2^{2p+1}} \left( \frac{1}{\sin^2(\pi/2^{p+1})} + \frac{1}{\sin^2 \big( \pi/ 2^{p+1}}) \right) \\
& = & \frac{1}{2^{2p}} \frac{1}{\sin^2(\pi/2^{p+1})} > \frac{1}{2^{2p}} \frac{1}{(\pi/2^{p+1})^2} = \frac{4}{\pi^2}.
\end{eqnarray*} 

Mind this is only half our chance of success, the other half coming from the case $\tilde l = 2^p(1- \tilde\omega)$. Thus, we have at least a probability of $8/\pi^2$ of obtaining $l_1$ or $l_2$ with an error less than $1/2^p$. It follows that, if we write 
$$
\Lambda = \left| \frac{\tilde l}{2^p} - \omega \right| \leq \frac{1}{2^p},
$$ 
then with probability of at least $8/\pi^2$,
\begin{eqnarray*}
|t - \tilde t| & = & 2^n \left| \sin^2(\pi \omega) - \sin^2 \left( \frac{\pi \tilde l}{2^p} \right) \right| \\
& = & 2^n | \sin^2(\pi \omega) - \sin^2 (\pi \omega \pm \pi \Lambda) | \\
& = & 2^n | \sin^2(\pi \omega) - (\sin(\pi \omega)\cos(\pi\Lambda)\pm \sin(\pi\Lambda)\cos(\pi\omega))^2 | \\
& = & 2^n \big| \sin^2(\pi \omega) - \big( \sin^2(\pi \omega)\cos^2(\pi\Lambda) + \sin^2(\pi\Lambda)\cos^2(\pi\omega) \pm \sin(2\pi\Lambda) \sin(\pi \omega) \cos(\pi\omega) \big) \big| \\
& = & 2^n \big| \sin^2(\pi\Lambda) (\sin^2(\pi \omega) - \cos^2(\pi\omega))  \pm \sin(2\pi\Lambda) \sin(\pi \omega) \cos(\pi\omega) \big) \big| \\
& \leq & 2^n |\sin(2\pi\Lambda)\sin(\pi\omega)\cos(\pi\omega)| + 2^n \sin^2(\pi\Lambda)|1-\sin^2(\pi\omega)| \\
& \leq & 2^{n+1}\pi\Lambda\sqrt{\frac{t}{2^n} \left( 1-\frac{t}{2^n} \right)} + 2^n (\pi\Lambda)^2 \left| 1-\frac{t}{2^n} \right| \\
& = & 2\pi\Lambda \sqrt{t(2^n-t)} + (\pi\Lambda)^2 \left| 2^n - t \right|. \\
\end{eqnarray*}
This finishes the proof, as $\Lambda \leq 1/2^p$.
\end{proof}

\newpage

\bibliography{refs}
\bibliographystyle{siam} 

\end{document}